\documentclass{article}[11pt]
\usepackage{fullpage}
\usepackage{amsthm}
\usepackage{authblk}
\usepackage{thmtools,thm-restate}

\declaretheorem[name=Theorem]{theorem}
\declaretheorem[name=Construction]{construction}
\declaretheorem[name=Definition]{definition}
\declaretheorem[name=Corollary]{corollary}
\declaretheorem[name=Lemma]{lemma}

\declaretheorem[name=Proposition]{proposition}
  
\usepackage{microtype}
\usepackage{amsmath,amssymb}
\usepackage{xspace}
\usepackage{dsfont}
\usepackage[colorlinks=true,urlcolor=webblue,linkcolor=webgreen,filecolor=webblue,citecolor=webgreen,pdfpagemode=UseOutlines,pdfstartview=FitH,pdfpagelayout=OneColumn,bookmarks=true]{hyperref}
\usepackage{enumerate}
\usepackage{enumitem}
\usepackage{xcolor}
\usepackage{breakcites}

\definecolor{mgreen}{rgb}{.1,.7,0}

\usepackage{algorithm, algpseudocode, float}

\usepackage[n,keys,asymptotics]{cryptocode}

\newcommand{\N}{\mathbb{N}}

\newcommand{\R}{\mathbb{R}}
\newcommand{\C}{\mathbb{C}}
\newcommand{\ket}[1]{|#1\rangle} %
\newcommand{\bra}[1]{\langle #1|} %
\newcommand{\tr}[1]{\mathrm{Tr}\left(#1\right)} %

\newcommand\calH{\mathcal{H}}
\newcommand\calO{\mathcal{O}}

\newcommand\veps{\varepsilon}

\newcommand{\one}{\mathds 1}

\newcommand{\bad}{\textsf{BAD}}

\definecolor{webgreen}{rgb}{0,.5,0}
\definecolor{webblue}{rgb}{0,0,.5}
\newcommand\algo{\mathcal}
\newcommand{\KeyGen}{\ensuremath{\mathsf{KeyGen}}}

\newcommand{\Mac}{\ensuremath{\mathsf{Mac}}}
\newcommand{\Ver}{\ensuremath{\mathsf{Ver}}}

\newcommand{\Sign}{\ensuremath{\mathsf{Sign}}}

\newcommand{\expref}[2]{\texorpdfstring{\hyperref[#2]{#1~\ref{#2}}}{#1~\ref{#2}}}
\newcommand{\bits}{\{0,1\}}
\newcommand{\inrand}{\in_R}
\newcommand{\from}{\leftarrow}
\newcommand{\BU}{\ensuremath{\mathsf{BU}}\xspace}
\newcommand{\mBU}{\ensuremath{\mathsf{mBU}}\xspace}
\newcommand{\sBU}{\ensuremath{\mathsf{sBU}}\xspace}
\newcommand{\EUFCMA}{\ensuremath{\mathsf{EUF\mbox{-}CMA}}\xspace}
\newcommand{\SUFCMA}{\ensuremath{\mathsf{SUF\mbox{-}CMA}}\xspace}
\newcommand{\BZ}{\ensuremath{\mathsf{PO}}\xspace}
\newcommand{\GYZ}{\ensuremath{\mathsf{GYZ}}\xspace}
\newcommand{\PRF}{\ensuremath{\mathsf{PRF}}\xspace}
\newcommand{\qPRF}{\ensuremath{\mathsf{qPRF}}\xspace}
\newcommand{\pkreg}{\ensuremath{P\!K}\xspace}
\newcommand{\skreg}{\ensuremath{S\!K}\xspace}
\newcommand{\bfilter}{Bernoulli-preserving hash\xspace}
\newcommand{\blindforge}{\ensuremath{\mathsf{BlindForge}}\xspace}
\newcommand{\acc}{\ensuremath{\mathsf{acc}}}
\newcommand{\rej}{\ensuremath{\mathsf{rej}}}
\newcommand{\win}{\ensuremath{\mathsf{win}}}
\newcommand{\supp}{\ensuremath{\mathsf{supp}}}

\newcommand{\outerprod}[2]{|#1\rangle\langle #2|}

\newcommand{\LWE}{\ensuremath{\mathsf{LWE}}\xspace}

\DeclareMathOperator*{\Exp}{\mathbb{E}}

\newcommand{\proj}[1]{\ensuremath{|#1\rangle \langle #1|}}
\newcommand{\adver}{\ensuremath{\mathcal{A}}\xspace}
\newcommand{\ketbra}[2]{\left|#1\right\rangle\!\!\left\langle #2\right|}
\newcommand{\braket}[2]{\left\langle #1 \mid #2 \right\rangle}
\newcommand{\Hi}{\mathcal{H}}
\newcommand{\Tr}{\mathrm{Tr}}
\newcommand{\CNOT}{\mathrm{CNOT}}

\newif\ifsubmission
\submissionfalse
\ifsubmission
\newcommand{\fs}[1]{}
\newcommand{\ga}[1]{}
\newcommand{\cm}[1]{}
\else
\newcommand{\fs}[1]{[\textcolor{blue}{FS: {#1}}]}
\newcommand{\ga}[1]{{\noindent \textcolor{purple}{\emph{(GA:  #1)}}}{}}
\newcommand{\cm}[1]{{\noindent \textcolor{mgreen}{\emph{(CM:  #1)}}}{}}
\fi

\usepackage{subcaption}

\title{Quantum-access-secure message authentication\\ via blind-unforgeability}

\author[1]{Gorjan Alagic}
\author[2]{Christian Majenz}
\author[3]{Alexander Russell}
\author[4]{Fang Song}

\affil[1]{{\small QuICS, University of Maryland, and NIST, Gaithersburg, Maryland}}
\affil[2]{{\small QuSoft and Centrum Wiskunde \& Informatica, Amsterdam}}
\affil[3]{{\small Department of Computer Science and Engineering,
    University of Connecticut}} \affil[4]{{\small Department of
    Computer Science, Portland State University}}

\begin{document}

\maketitle

\begin{abstract}

  Formulating and designing authentication of classical messages in
  the presence of adversaries with quantum query access has been a
  longstanding challenge, as the familiar classical notions of
  unforgeability do not directly translate into meaningful notions in
  the quantum setting. A particular difficulty is how to fairly
  capture the notion of ``predicting an unqueried value'' when the
  adversary can query in quantum superposition.

  We propose a natural definition of unforgeability against quantum
  adversaries called \emph{blind unforgeability}. This notion defines
  a function to be predictable if there exists an adversary who can
  use ``partially blinded'' oracle access to predict values in the
  blinded region. We support the proposal with a number of technical
  results. We begin by establishing that the notion coincides with
  EUF-CMA in the classical setting and go on to demonstrate that the
  notion is satisfied by a number of simple guiding examples, such as
  random functions and quantum-query-secure pseudorandom functions. We
  then show the suitability of blind unforgeability for supporting
  canonical constructions and reductions. We prove that the
  ``hash-and-MAC'' paradigm and the Lamport one-time digital signature
  scheme are indeed unforgeable according to the definition. To
  support our analysis, we additionally define and study a new variety
  of quantum-secure hash functions called \emph{Bernoulli-preserving}.

Finally, we demonstrate that blind unforgeability is
stronger than a previous definition of Boneh and Zhandry [EUROCRYPT
'13, CRYPTO '13] in the sense that we can construct an explicit function family
which is forgeable by an attack that is recognized by blind-unforgeability, yet satisfies the definition by Boneh and Zhandry.

\end{abstract}
\vspace{1cm}

\begin{center}
	\fbox{\begin{minipage}{.98\textwidth}
			Note: An earlier version of this article contained a theorem that the new security notion ``blind-unforgeability'' (\BU) we introduce implies the notion of ``plus-one'' unforgeability (\BZ) that was hitherto the only proposed generalization of \EUFCMA to $>1$ quantum chosen-message queries. Unfortunately the proof contained an error. We thank Shih-Han Hung for discovering the error. We have now removed the claim and the question whether the implication holds is currently open. We would like to emphasize that the example presented in Section \ref{sec:BZisbroken} disqualifies \BZ as a quantum-access generalization of \EUFCMA, so currently there are two reasonable possible definitions of  ``quantum-access \EUFCMA'': the notion of blind-unforgeability proposed in this article, and its conjunction with plus-one unforgeability.
	\end{minipage}}	
\end{center}

\section{Introduction}

Large-scale quantum computers will break widely-deployed public-key
cryptography, and may even threaten certain post-quantum
candidates~\cite{Shor97,NIST16,CDPR16,EHKS14,BS16}. %
Even elementary symmetric-key constructions like Feistel ciphers and
CBC-MACs become vulnerable in quantum attack models where the
adversary is presumed to have quantum query access to some part of the
cryptosystem~\cite{KM10,KM12,KLLN16,SS17}. As an example, consider
encryption in the setting where the adversary has access to the
unitary operator $\ket{x}\ket{y} \mapsto \ket{x}\ket{y \oplus f_k(x)}$,
where $f_k$ is the encryption or decryption function with secret key
$k$. While it is debatable if this model reflects physical
implementations of symmetric-key cryptography, it appears necessary in
a number of generic settings, such as public-key encryption and
hashing with public hash functions. It could also be relevant when
private-key primitives are composed in larger protocols, e.g., by
exposing circuits via obfuscation~\cite{SW14}. Setting down
appropriate security definitions in this quantum attack model is the
subject of several threads of recent research~\cite{BZ13,GHS16}.

In this article, we study authentication of classical information in
the quantum-secure model. Here, the adversary is granted quantum query
access to the signing algorithm of a message authentication code (MAC)
or a digital signature scheme, and is tasked with producing valid
forgeries. In the purely classical setting, we insist that the
forgeries are fresh, i.e., distinct from previous queries to the
oracle. %
When the function may be queried in superposition, however, it's
unclear how to meaningfully reflect this constraint that a forgery was
previously
``unqueried.'' %
For example, it is clear that an adversary that simply queries with a
uniform superposition and then measures a forgery---a feasible attack
against any function---should not be considered successful. On the
other hand, an adversary that uses the same query to discover some
structural property (e.g., a superpolynomial-size period in the MAC)
should be considered a break. Examples like these indicate the
difficulty of the problem. How do we correctly ``price'' the queries?
How do we decide if a forgery is fresh? Furthermore, how can this be
done in a manner that is consistent with these guiding examples?  In
fact, this problem has a natural interpretation that goes well beyond
cryptography: \emph{What does it mean for a classical function to
  appear unpredictable to a quantum oracle algorithm?}\footnote{The
  related notion of ``appearing \emph{random} to quantum oracle
  algorithms'' has a satisfying definition, which can be fulfilled
  efficiently~\cite{Zhandry12}.}

\subsection{Previous approaches}
The first approach to this problem was suggested by Boneh and
Zhandry~\cite{BZ13a}.  %
They define a MAC to be unforgeable if, after making $q$ queries to
the MAC, no adversary can produce $q+1$ valid input-output pairs
except with negligible probability. We will refer to this notion as
``\BZ security'' (\BZ for ``plus one,'' and $k$-\BZ when the adversary
is permitted a maximum of $k$ queries). Among a number of results,
Boneh and Zhandry prove that this notion can be realized by a
quantum-secure pseudorandom function (\qPRF).

Another approach, due to Garg, Yuen and Zhandry~\cite{GYZ17} (\GYZ),
considers a function \emph{one-time} unforgeable if only a trivial
``query, measure in computational basis, output result''
attack\footnote{Technically, the \emph{Stinespring dilation} \cite{Stinespring1955} of a
  computational basis measurement is the most general attack.} is
allowed. Unfortunately, it is not clear how to extend \GYZ to two or
more queries. Furthermore, the single query is allowed in a limited
query model with an non-standard restriction.\footnote{Compared to the
  standard quantum oracle for a classical function, GYZ require the
  output register to be empty prior to the query.} Zhandry recently
showed a separation between \BZ and \GYZ by means of the
powerful tool of obfuscation~\cite{Zhandry17}.

It is interesting to note that similar problems arise in encryption
schemes of \emph{quantum} data and a convincing solution was recently
found~\cite{AGM18,SignCrypt}. However, it relies on the fact that for
quantum messages, \emph{authentication implies secrecy}. This enables
``tricking'' the adversary by replacing their queries with ``trap''
plaintexts to detect replays. As unforgeability and secrecy are
orthogonal in the classical world, adversaries would easily recognize
the spoofed oracle. This renders the approach
of~\cite{AGM18,SignCrypt} inapplicable in this case.

\subsection{Unresolved issues} 

\BZ security, the only candidate definition of quantum-secure
unforgeability in the general, multi-query setting, appears to be
insufficient for several reasons. First, as observed in~\cite{GYZ17},
it is a priori unclear if \BZ security rules out forging on a message
region $A$ while making queries to a signing oracle supported on a
disjoint message region $B$. Second, there may be unique features of
quantum information, such as the destructiveness of quantum
measurement, which \BZ does not capture. In particular, quantum
algorithms must sometimes ``consume'' (i.e., fully measure) a state to
extract some useful information, such as a symmetry in the
oracle. There might be an adversary that makes one or more quantum
queries but then must consume the post-query states completely in
order to make a single, but convincing, forgery.

Surprisingly, prior to this work none of these plausible attack strategies have been exploited to give a separation between \BZ and ``intuitive security.''

\section{Summary of results}

\subsection{A new definition: Blind-unforgeability}

To address the abovementioned issues, and in light of the concrete
``counterexample'' presented below as
\expref{Construction}{con:real-BZ-killer-summary}, we develop a new
definition of many-time unforgeability we call
``blind-unforgeability'' (or \BU). In this approach we examine the
behavior of adversaries in the following experiment. The adversary is
granted quantum oracle access to the MAC, ``blinded'' at a random
region $B$. Specifically, we set $B$ to be a random
$\epsilon$-fraction of the message space, and declare that the oracle
function will output $\bot$ on all of $B$.
\[
  B_\epsilon \Mac_k (x) := \begin{cases} 
    \bot & \text{if $x \in B_\epsilon$},\\
    \Mac_k (x) & \text{otherwise.}
  \end{cases}
\]
Given a MAC $(\Mac, \Ver)$, an adversary $\algo A$, and $\algo A$-selected parameter $\epsilon$, the ``blind forgery experiment'' is:
\begin{enumerate}
\item Generate key $k$ and random blinding $B_\epsilon$; 
\item Produce candidate forgery $(m, t) \from \algo A^{B_\epsilon \Mac_k} (1^n)$.
\item Output $\win$ if $\Ver_k(m, t) = \acc$ and $m \in B_\epsilon$; otherwise output $\rej$.
\end{enumerate}
\begin{definition} A MAC is blind-unforgeable \emph{(\BU)} if for every adversary $(\algo A, \epsilon)$, the probability of winning the blind forgery experiment is negligible.
\end{definition}

In this work, \BU will typically refer to the case where $\algo A$ is an efficient quantum algorithm (QPT) and the oracle is quantum, i.e., $\ket{x}\ket{y} \mapsto \ket{x}\ket{y \oplus B_\epsilon \Mac_k(x)}$. We will also consider $q$-\BU, the information-theoretic variant where the total number of queries is a priori fixed to $q$. We remark that the above definition is also easy to adapt to other settings, e.g., classical security against PPT adversaries, quantum or classical security for digital signatures, etc. 

We remark that one could define a variant of the above where the adversary is allowed to describe the blinding distribution, rather than it being uniform. However, this is not a stronger notion. By a straightforward argument, an adversary wins in the chosen-blinding \BU game if and only if it wins with a uniform $\epsilon$-blinding for inverse-polynomial $\epsilon$. Indeed, the adversary can just simulate its chosen blinding herself, and this still succeeds with inverse polynomial probability when interacting with a standard-blinded oracle (see \expref{Theorem}{thm:simulation} below).

\subsection{Results about blind-unforgeability}

To solidify our confidence in the new notion, we collect a series of
results which we believe establish \BU as a definition of
unforgeability that captures the desired intuitive security
requirement. In particular, we show that
\BU %
classifies a host of representative examples (in fact, all examples
examined thus far) as either forgeable or unforgeable in a way that
agrees with our cryptographic intuition. We show that \BZ certifies
certain MACs as secure but are completely broken by a quantum-access
attack in a \emph{strong and intuitive sense}
(Section~\ref{sec:intro-bz-killer} below). While we cannot currently
prove that \BU is strictly stronger than \BZ as a security notion, we
can show that \BU implies a weaker version of \BZ in
Section~\ref{sec:bu2qbz}.

\paragraph{Relations and characterizations.}

One key technical ingredient that informs our intuition and characterization results about \BU is a general simulation theorem, which tightly controls the deviation in the behavior of an algorithm when subjected to the \BU experiment. 

\begin{theorem}\label{thm:simulation}
	Let $\algo A$ be a quantum query algorithm making at most $T$ queries. Let $f:X \to Y$ be a function, $B_\epsilon$ a random $\epsilon$-blinding subset of $X$, and for each $B\subset X$, let $g_B$ be a function with support $B$. Then
	$$
	\Exp_{B_\epsilon}\left\|\algo A^f(1^n) - \algo A^{f \oplus g_{B_\epsilon}}(1^n)\right\|_1 \leq 2T\sqrt{\epsilon}\,.
	$$
      \end{theorem}
      
      \label{sec:bu-explain} This result can be viewed as strong
      evidence that algorithms that produce ``good forgeries'' in any
      reasonable sense %
      will not be disturbed too much by blinding, and will thus also
      win the \BU experiment. We can formulate and prove this
      intuition explicitly for a wide class of adversaries, as
      follows. Given an oracle algorithm $\algo A$, we let
      $\supp(\algo A)$ denote the union of the supports of all the
      queries of $\algo A$, taken over all choices of oracle function.

\begin{theorem}[informal]\label{intro:thm-support}
	Let $\algo A$ be QPT and $\supp(\algo A) \cap R = \emptyset$ for
	some $R \neq \emptyset$. Let $\Mac$ be a MAC, and suppose
	$\algo A^{\Mac_k}(1^n)$ outputs a valid pair $(m, \Mac_k(m))$ with
	$m \in R$ with noticeable probability. Then $\Mac$ is not \BU
	secure.
\end{theorem}

\paragraph{Blind-unforgeable MACs.}

Next, we show that several natural constructions satisfy \BU. We first show that a random function is blind-unforgeable.

\begin{theorem}
Let $R : X \to Y$ be a random function such that $1/|Y|$ is negligible. Then $R$ is a blind-unforgeable MAC.
\end{theorem}

This together with results of Zhandry~\cite{Zhandry12} and Boneh and
Zhandry~\cite{BZ13a} leads to efficient \BU-secure constructions.

\begin{corollary}\label{cor:qprfandk-wise}
Quantum-secure pseudorandom functions \emph{(\qPRF)} are \BU-secure MACs, and (4q+1)-wise independent functions are $q$-\BU-secure MACs.
\end{corollary}

We can then invoke a recent result about the quantum-security of
domain-extension schemes such as NMAC and HMAC~\cite{SY17}, and obtain
variable-length \BU-secure MACs from any \qPRF.

In the setting of public verification, we show that the one-time Lamport signature scheme \cite{Lam79} is \BU-secure, provided that the underlying hash function family $\mathcal R: X \to Y$ is modeled as a random oracle. %

\begin{theorem}\label{thm:lamport-summary}
	Let $\mathcal R : X \to Y$ be a random function family. Then the Lamport scheme $L_\mathcal R$ is \BU against adversaries which make one quantum query to $L_\mathcal R$ and poly-many quantum queries to $\mathcal R$.
\end{theorem}

\paragraph{Hash-and-MAC.} 

Consider the following natural variation on the blind-forgery
experiment. To blind $F : X \to Y$, we first select a hash function
$h : X \to Z$ and a blinding set $B_\epsilon \subseteq Z$; we then
declare that $F$ will be blinded on $x \in X$ whenever
$h(x) \in B_\epsilon$. We refer to this as ``hash-blinding.'' We call
$h$ a \bfilter if, for every oracle function $F$, no QPT can
distinguish between an oracle that has been hash-blinded with $h$, and
an oracle that has been blinded in the usual sense. Recall the notion
of \emph{collapsing} from \cite{Unruh16a}.

\begin{theorem}
Let $h : X \to Y$ be a hash function. If $h$ is \bfilter, then it is also collapsing. Moreover, against adversaries with classical oracle access, $h$ is a \bfilter if and only if it is collision-resistant.
\end{theorem}

We apply this new notion to show security of the Hash-and-MAC construction $\Pi^h = (\Mac^h, \Ver^h)$ with $\Mac^h_k (m): = \Mac_k(h(m))$.

\begin{theorem}
Let $\Pi = (\Mac_k, \Ver_k)$ be a \BU-secure MAC with $\Mac_k : X \to Y$, and let $h: Z \to X$ a \bfilter. Then $\Pi^h$ is a \BU-secure MAC.
\end{theorem}

We also show that the Bernoulli-preserving property can be satisfied by pseudorandom constructions, as well as a (public-key) hash based on \emph{lossy functions} from \LWE~\cite{PW08,UnruhAC16}.

\subsection{A concrete ``counterexample''  for  \BZ}\label{sec:intro-bz-killer}

Supporting our motivation to devise a new unforgeability definition, we present a construction of a MAC which is forgeable (in a
strong intuitive sense) and yet is classified by \BZ as secure.

\begin{construction}\label{con:real-BZ-killer-summary}
		Given $k = (p, f, g,h)$ where $p\in\bits^{n}$ is a random period and $f,g,h: \bits^{n}\to\bits^{n}$ are random functions, define $M_k : \bits^{n+1} \to \bits^{2n}$ by
		\vspace{-.01cm}
		$$
		M_k(x) =
		\begin{cases}
		g(x' \bmod p)\|f(x')& x=1\|x'\,,\\
		0^n\|h(x')& x=0\|x', \ x'\neq p\,,\\
		0^{2n}& x=0\|p\,.
		\end{cases}
		$$
\end{construction}
Define $g_p(x) := g(x \bmod p)$ and consider an adversary that queries only on messages starting with $1$, as follows:
\begin{equation}\label{eq:attack}
\sum_{x, y} \ket{1, x}_X \ket{0^n}_{Y_1} \ket{y}_{Y_2}
\longmapsto 
\sum_{x, y} \ket{1, x}_X \ket{g_p(x)}_{Y_1} \ket{y \oplus f(x)}_{Y_2}\,;
\end{equation}
discarding the first qubit and $Y_2$ then yields
$\sum_x \ket{x}\ket{g_p(x)}$, as $\sum_{ y}  \ket{y \oplus f(x)}_{Y_2}=\sum_{ y}  \ket{y }_{Y_2}$. One can then recover $p$ via
period-finding and output
$(0\|p,0^{2n})$. %
We emphasize that the forgery was queried with \emph{zero}
amplitude. In practice, we can interpret it as, e.g., the attacker
queries only on messages starting with ``\textsf{From: Alice}'' and
then forges a message starting with ``\textsf{From: Bob}''. Despite
this, we can show that it is \BZ-secure.

\begin{theorem}
	The family $M_k$ (for uniformly random $k = (p,f,g,h)$) is \BZ-secure.
\end{theorem}

The \BZ security of $M$ relies on a dilemma the adversary faces at each
query: either learn an output of $f$, or obtain a superposition of
$(x, g(x))$-pairs for Fourier sampling. Our proof shows that, once the
adversary commits to one of these two choices, the other option is
irrevocably lost. Our result can thus be understood as a refinement of
an observation of Aaronson: quantumly learning a property sometimes requires \emph{uncomputing} some
information~\cite{Aaronson2003}. Note that, while Aaronson could rely
on standard (asymptotic) query complexity techniques, our problem is
quite fragile: \BZ security describes a task which should be hard
with $q$ queries, but is completely trivial given $q+1$ queries. Our
proof makes use of  a new quantum random
oracle technique of Zhandry~\cite{Zhandry2018}.

A straightforward application of \expref{Theorem}{intro:thm-support} shows that \expref{Construction}{con:real-BZ-killer-summary} is \BU-insecure. In particular, we have the following.

\begin{corollary}\label{cor:separation}
	There exists a \BZ-secure MAC which is \BU-insecure.
\end{corollary}
The relationship between \BU, \BZ and a few other notions are
visualized in \expref{Figure}{fig:rel}.

\begin{figure}[h]
	\vskip -3ex
	\centering
	\begin{subfigure}{.49\linewidth}
		\centering
		\[
		\begin{matrix}
		\EUFCMA & \overset{\cite{BZ13a}}{\Longleftrightarrow} & \BZ &
		\overset{\text{\expref{Proposition}{prop:MAC-EUF=BU}}}{\Longleftrightarrow} & \BU
		\end{matrix}  
		\]
		\caption*{Unforgeability against classical
			adversaries}\label{fig:relationclassical}
	\end{subfigure}
	\hfill
	\begin{subfigure}{.49\linewidth}
		\centering
		\[
		\begin{matrix}
		\BZ &
		\overset{\text{\expref{Corollary}{cor:separation}}}{\not\Longrightarrow}%
		& \BU &
		\substack{\overset{\text{Observation}}{\not\Longrightarrow} \\ {\underset{\text{\expref{Corollary}{cor:qprfandk-wise}}}{\Longleftarrow}}}&
		\qPRF
		\end{matrix}  
		\]
		\caption*{Unforgeability against quantum
			adversaries}\label{fig:relationquantum}
	\end{subfigure}
	\caption{Relationship between different unforgeability notions}\label{fig:rel}
	\hfill
	\vskip -3ex
\end{figure}

\subsection{Acknowledgements}

The authors thanks Shih-Han Hung for discovering an error in an earlier version of this paper. CM thanks Ronald de Wolf for helpful discussions on query
complexity. 
GA acknowledges support from the NSF under grant CCF-1763736, from the U.S. Army Research Office under Grant Number W911NF-20-1-0015, and from the U.S. Department of Energy under Award Number DE-SC0020312. 
CM was funded by a NWO VIDI grant (Project No.  639.022.519) and a NWO VENI grant (Project No. VI.Veni.192.159). 
FS acknowledges support from NSF grant CCF-2042414. 
AR acknowledges support from NSF grant CCF-1763773.

\section{Preliminaries}\label{sec:prelims}

\paragraph{Basic notation and conventions.}
Given a finite set $X$, the notation $x \inrand X$ will mean that $x$ is a uniformly random element of $X$. Given a subset $B$ of a set $X$, let $\chi_B : X \to \bits$ denote the characteristic function of $B$, i.e., $\chi_B(x) = 1$ if $x \in B$ and $\chi_B(x)=0$ otherwise. %
When we say that a classical function $F$ is efficiently computable, we mean that there exists a uniform family of deterministic classical circuits which computes $F$. We will consider three classes of algorithms: (i.) unrestricted algorithms, modeling computationally unbounded adversaries, (ii.) probabilistic poly-time algorithms (PPTs), modeling classical adversaries, and (iii.) quantum poly-time algorithms (QPTs), modeling quantum adversaries. We assume that the latter two are given as polynomial-time uniform families of circuits. For PPTs, these are probabilistic circuits. For QPTs, they are quantum circuits, which may contain both unitary gates and measurements. We will often assume (without loss of generality) that the measurements are postponed to the end of the circuit, and that they take place in the computational basis. Given an algorithm $\algo A$, we let $\algo A(x)$ denote the (in general, mixed) state output by $\algo A$ on input $x$. In particular, if $\algo A$ has classical output, then $\algo A(x)$ denotes a probability distribution. Unless otherwise stated, the probability is taken over all random coins and measurements of $\algo A$ and any randomness used to select the input $x$. If $\algo A$ is an oracle algorithm and $F$ a classical function, then $\algo A^F(x)$ is the mixed state output by $\algo A$ equipped with oracle $F$ and input $x$; the probability is now also taken over any randomness used to generate $F$. 

We will distinguish between two ways of presenting a function $F: \bits^n \to \bits^m$ as an oracle. First, the usual ``classical oracle access'' simply means that each oracle call grants one classical invocation $x \mapsto F(x)$. This will always be the oracle model for PPTs. Second, ``quantum oracle access'' will mean that each oracle call grants an invocation of the $(n+m)$-qubit unitary gate 
$
\ket{x}\ket{y} \mapsto \ket{x}\ket{y \oplus F(x)}\,.
$
This will always be the oracle model for QPTs. Note that both QPTs and unrestricted algorithms could in principle receive either oracle type.

We will need the following lemma. We use the formulation from \cite[Lemma 2.1]{BZ13}, which is a special case of a more general ``pinching lemma'' of Hayashi~\cite{H02}.
\begin{lemma}\label{lem:add-measure}
Let $\algo A$ be a quantum algorithm and $x \in \bits^*$. Let $\algo A_0$ be another quantum algorithm obtained from $\algo A$ by pausing $\algo A$ at an arbitrary stage of execution, performing a partial measurement that obtains one of $k$ outcomes, and then resuming $\algo A$. Then $\Pr[\algo A_0(1^n) = x] \geq \Pr[\algo A(1^n) = x]/k$.
\end{lemma}

We denote the trace distance between states $\rho$ and $\sigma$ by $\delta(\rho, \sigma)$. Recall that this is simply half the trace norm of the difference, i.e., $\delta(\rho, \sigma) = ({1}/{2})\|\rho - \sigma\|_1$. When $\rho$ and $\sigma$ are classical probability distributions, the trace distance is equal to the total variation distance. 

\subsection{Quantum-secure pseudorandomness}
A quantum-secure pseudorandom function (\qPRF) is a family of classical, deterministic, efficiently-computable functions which appear random to QPT adversaries with quantum oracle access.
\begin{definition}\label{def:qPRF}
An efficiently computable function family $f : K \times X \to Y$ is a quantum-secure pseudorandom function \emph{($\qPRF$)} if, for all QPTs $\algo D$,
\[
\Bigl|\underset{k \inrand K}{\Pr}\bigl[\mathcal D^{f_k}(1^n) = 1\bigr] - 
\underset{g \inrand \mathcal F_X^Y}{\Pr}\bigl[\mathcal D^g(1^n) = 1\bigr] \Bigr|\,\le \mathrm{negl}(n)\,.
\]
\end{definition}
Here $\mathcal F_X^Y$ denotes the set of all functions from $X$ to $Y$. The standard ``GGM+GL'' construction of a \PRF yields a \qPRF when instantiated with a quantum-secure one-way function~\cite{Zhandry12}. One can also construct a \qPRF directly from the Learning with Errors assumption~\cite{Zhandry12}. If we have an a priori bound on the number of allowed queries, then a computational assumption is not needed. 

\begin{theorem}[Lemma 6.4 in~\cite{BZ13a}] \label{thm:k-wise}
Let $q, c \geq 0$ be integers, and $f: K \times X \to Y$ a $(2q + c)$-wise independent family of functions. Let $\algo D$ be an algorithm making no more than $q$ quantum oracle queries and $c$ classical oracle queries. Then
\[
\underset{k \inrand K}{\Pr}\bigl[\mathcal D^{f_k}(1^n) = 1\bigr] = 
\underset{g \inrand \mathcal F_X^Y}{\Pr}\bigl[\mathcal D^g(1^n) = 1\bigr]\,.
\]
\end{theorem}

\subsection{\BZ-unforgeability}

Boneh and Zhandry define unforgeability (against quantum queries) for classical MACs as follows~\cite{BZ13a}. They also show that random functions satisfy this notion.

\begin{definition}\label{def:BZ}
Let $\Pi = (\KeyGen, \Mac, \Ver)$ be a MAC with message set $X$. Consider the following experiment with an algorithm $\algo A$:
\begin{enumerate}
\item \emph{Generate key:} $k \from \KeyGen(1^n)$.
\item \emph{Generate forgeries:} $\algo A$ receives quantum oracle for $\Mac_k$, makes $q$ queries, and outputs a string $s$;
\item \emph{Outcome:} output \textsf{win} if $s$ contains $q+1$ distinct input-output pairs of $\Mac_k$, and \textsf{fail} otherwise.
\end{enumerate}
We say that $\Pi$ is \BZ-secure if no adversary can succeed at the above experiment with better than negligible probability.
\end{definition}

\subsection{The Fourier Oracle}\label{sec:superoracles}

Our separation proof will make use of a new technique of Zhandry~\cite{Zhandry2018} for analyzing random oracles. We briefly describe this framework.

A random function $f$ from $n$ bits to $m$ bits can be viewed as the outcome of a quantum measurement. More precisely, let $\Hi_F=\bigotimes_{x\in\bits^n}\Hi_{F_x}$, where $\Hi_{F_x}\cong\C^{2^m}$. Then set
$f(x)\leftarrow \mathcal M_{F_x}(\eta_F)$, where
\[
  \eta_F=\proj{\phi_0}^{\otimes 2^n}, \qquad \ket{\phi_0}=2^{-\frac{m}{2}}\sum_{y\in\bits^m}\ket y\,,
\]
and $\mathcal M_{F_x}$ denotes the measurement of the register $F_x$
in the computational basis. This measurement commutes with any
$\CNOT_{A:B}$ gate with control qubit $A$ in $F_x$ and target qubit
$B$ outside $F_x$. It follows that, for any quantum algorithm making
queries to a random oracle, the output distribution is identical if
the algorithm is instead run with the following oracle:
\begin{enumerate}
	\item Setup: Prepare the state $\eta_F$.
	\item Upon a query with query registers $X$ and $Y$, controlled on $X$ being in state $\ket x$, apply $(\CNOT^{\otimes m})_{F_x:Y}$. 
	\item After the algorithm has finished, measure $F$ to determine the success of the computation.
\end{enumerate}

We denote the oracle unitary defined in step 2 above by $U^{\mathrm O}_{XYF}$.
Having defined this oracle representation, we are free to apply any unitary $U_H $ to the oracle state, so long as we then also apply the conjugated query unitary
\[
  U_H(\CNOT^{\otimes m})_{F_x:Y}U_H^\dagger
\]
in place of $U^{\mathrm O}_{XYF}$. We choose $U_H=H^{\otimes m2^n}$,
which means that the oracle register starts in the all-zero state
now. Applying Hadamard to both qubits reverses the direction of CNOT,
i.e.,
\[
  H_A\otimes H_B\CNOT_{A:B}H_A\otimes H_B=\CNOT_{B:A}\,,
\]
so the adversary-oracle-state after a first query with query state
$\ket x_X\ket{\phi_y}_Y$ is
\begin{equation}\label{eq:example1}
	\ket x_X\ket{\phi_y}_Y\ket{0^m}^{\otimes 2^n}
	\longmapsto
	\ket x_X\ket{\phi_y}_Y\ket{0^m}^{\otimes (\textsf{lex}(x)-1)}\ket{y}_{F_x}\ket{0^m}^{\otimes (2^n-\textsf{lex}(x))},
\end{equation}
where $\textsf{lex}(x)$ denotes the position of $x$ in the lexicographic ordering of $\{0,1\}^n$, and we defined the Fourier basis state $\ket{\phi_y}=H^{\otimes m}\ket y$. In the rest of this section, we freely change the order in which tensor products are written, and keep track of the tensor factors through the use of subscripts. This adjusted representation is called the \emph{Fourier oracle} (FO), and we denote its oracle unitary by
\[
  U^{\mathrm{FO}}_{XYF}=\left(H^{\otimes m2^n}\right)_FU^{\mathrm O}_{XYF}\left(H^{\otimes m2^n}\right)_F\,.
\]

An essential fact about the FO is that each query can only change the number of non-zero entries in the FO's register by at most one. To formalize this idea, we define the ``number operator''
\[
  N_F=\sum_{x\in\bits^n}(\one-\proj 0)_{F_x}\otimes \one^{\otimes (2^n-1)}\,.
\]
The number operator can also be written in its spectral decomposition,
\begin{align*}
	N_F&=\sum_{l=0}^{2^n}l P_l \qquad \text{where} \qquad P_l=\sum_{r\in S_l}\proj r\,,\\
	S_l&=\left\{r\in\left(\bits^m\right)^{2^n}\Big||\{x\in\bits^n|r_x\neq 0\}|=l \right\}.
\end{align*}
Note that the initial joint state of a quantum query algorithm and the
oracle (in the FO-oracle picture described above) is in the image of
$P_0$. The following fact is essential for working with the Fourier
Oracle; to avoid disrupting the flow of the article, the proof is
given in \expref{Appendix}{sec:proof-zhancrement}.

\begin{restatable}{lemma}{numberoplem}
  \label{lem:Zhancrement}
  The number operator satisfies
  $
  \bigl\|\left[N_F,U_{XYF}^{FO}\right]\bigr\|_\infty=1.
  $
  In particular, the joint state of a quantum query algorithm and the oracle after the $q$-th query is in the kernel of $P_l$ for all $l>q$.
\end{restatable}

\section{The new notion: Blind-Unforgeability}\label{sec:def}

\subsection{Formal definition}
For ease of exposition, we begin by introducing our new security notion in a form analogue to the standard notion of existential unforgeability under chosen-message attacks, \EUFCMA. We will also later show how to extend our approach to obtain a corresponding analogue of strong unforgeability. We begin by defining a ``blinding'' operation. Let $f: X \rightarrow Y$ and $B \subseteq X$. We let 
\[
  Bf(x) = \begin{cases} 
    \bot & \text{if $x \in B$},\\
    f(x) & \text{otherwise.}
  \end{cases}
\]
We say that $f$ has been ``blinded'' by $B$. In this context, we will be particularly interested in the setting where elements of $X$ are placed in $B$ independently at random with a particular probability $\epsilon$; we let $B_\epsilon$ denote this random variable. (It will be easy to infer $X$ from context so we do not reflect it in the notation.)

Next, we define a security game in which an adversary is tasked with using a blinded MAC oracle to produce a valid input-output pair in the blinded set.

\begin{definition}\label{def:MAC-BU-game} Let $\Pi = (\KeyGen, \Mac, \Ver)$ be a MAC with message set $X$. Let $\algo A$ be an algorithm, and $\epsilon : \N \to \R_{\geq 0}$ an efficiently computable function. The blind forgery experiment $\blindforge_{\algo A, \Pi}(n, \epsilon)$ proceeds as follows:
\begin{enumerate}
\item \emph{Generate key:} $k \from \KeyGen(1^n)$.
\item \emph{Generate blinding:} select $B_\epsilon \subseteq X$ by placing each $m$ into $B_\epsilon$ independently with probability $\epsilon(n)$.
\item \emph{Produce forgery:} $(m, t) \from \algo A^{B_\epsilon \Mac_k} (1^n)$.
\item \emph{Outcome:} output $1$ if $\Ver_k(m, t) = \acc$ and $m \in B_\epsilon$; otherwise output $0$.
\end{enumerate}
\end{definition}

We say that a scheme is blind-unforgeable if, for any efficient adversary, the probability of winning the game is negligible. The probability is taken over the choice of key, the choice of blinding set, and any internal randomness of the adversary. We remark that specifying an adversary requires specifying (in a uniform fashion) both the algorithm $\algo A$ and the blinding fraction $\epsilon$.

\begin{definition}\label{def:MAC-BU} A MAC $\Pi$ is blind-unforgeable (BU) if for every polynomial-time uniform adversary $(\algo A, \epsilon)$,
\[
\Pr \bigl[\blindforge_{\algo A, \Pi}(n, \epsilon(n)) = 1] \leq \negl\,.
\]
\end{definition}

We also define the ``$q$-time'' variant of the blinded forgery game, which is identical to \expref{Definition}{def:MAC-BU-game} except that the adversary is only allowed to make $q$ queries to $B_\epsilon \Mac_k$ in step (3). We call the resulting game $\blindforge^q_{\algo A, \Pi}(n, \epsilon)$, and give the corresponding definition of $q$-time security (now against computationally unbounded adversaries).

\begin{definition}\label{def:q-MAC-BU} A MAC $\Pi$ is $q$-time blind-unforgeable ($q$-BU) if for every $q$-query adversary $(\algo A, \epsilon)$, we have
\[
\Pr \bigl[\blindforge^q_{\algo A, \Pi}(n, \epsilon(n)) = 1] \leq \negl\,.
\]
\end{definition}

The above definitions are agnostic regarding the computational power of the adversary and the type of oracle provided. For example, selecting PPT adversaries and classical oracles in \expref{Definition}{def:MAC-BU} yields a definition of classical unforgeability; we will later show that this is equivalent to standard \EUFCMA. The main focus of our work will be on \BU against QPTs with quantum oracle access, and $q$-\BU against unrestricted adversaries with quantum oracle access.

\subsection{Some technical details}

We now remark on a few details in the usage of \BU. First, strictly speaking, the blinding sets in the
security games above cannot be generated efficiently. However, a
pseudorandom blinding set will suffice. Pseudorandom blinding sets can
be generated straightforwardly using an appropriate pseudorandom
function, such as a $\PRF$ against PPTs or a $\qPRF$ against QPT. A
precise description of how to perform this pseudorandom blinding is
given in the proof of \expref{Corollary}{cor:qPRF-MAC}. Note that simulating the blinding requires computing and uncomputing the random function, so we must make two quantum queries for each quantum query of the adversary. Moreover, verifying whether the forgery is in the blinding set at the end requires one additional classical query. This means that $(4q+1)$-wise independent
functions are both necessary and sufficient for generating blinding
sets for $q$-query adversaries (see~\cite[Lemma 6.4]{BZ13a}). In any
case, an adversary which behaves differently in the random-blinding
game versus the pseudorandom-blinding game immediately yields a
distinguisher against the corresponding pseudorandom function.

\vspace{5pt}\noindent\emph{The blinding symbol.} There is some flexibility in how one defines the blinding symbol $\bot$. In situations where the particular instantiation of the blinding symbol might matter, we will adopt the convention that the blinded version $Bf$ of $f : \bits^n \to \bits^\ell$ is defined by setting $Bf : \bits^n \to \bits^{\ell+1}$, where $Bf(m) = 0^\ell || 1$ if $m \in B$ and $Bf(m) = f(m) || 0$ otherwise. One advantage of this convention (i.e., that $\bot = 0^\ell || 1$) is that we can compute on and/or measure the blinded bit (i.e., the $(\ell+1)$-st bit) without affecting the output register of the function. This will also turn out to be convenient for uncomputation.

\vspace{5pt}\noindent\emph{Strong blind-unforgeability.}\label{sec:strong}
The security notion \BU given in \expref{Definition}{def:MAC-BU} is an analogue of simple unforgeability, i.e., \EUFCMA, for the case of a quantum-accessible MAC/Signing oracle. It is, however, straightforward to define a corresponding analogue of strong unforgeability, i.e., \SUFCMA, as well.

The notion of strong blind-unforgeability, \sBU, is obtained by a  simple adjustment compared to \BU: we blind (message, tag) pairs rather than just messages. We briefly describe this for the case of MACs. Let $\Pi = (\KeyGen, \Mac, \Ver)$ be a MAC with message set $M$, randomness set $R$ and tag set $T$, so that $\Mac_k : M \times R \to T$ and $\Ver_k : M \times T \to \{\acc, \rej\}$ for every $k \from \KeyGen$. Given a parameter $\epsilon$ and an adversary $\algo A$, the strong blind forgery game proceeds as follows:
\begin{enumerate}
\item Generate key: $k \from \KeyGen$; generate blinding: select $B_\epsilon \subseteq M \times T$ by placing pairs $(m, t)$ in $B_\epsilon$ independently with probability $\epsilon$;
\item Produce forgery: produce $(m, t)$ by executing $\algo A(1^n)$ with quantum oracle access to the function
\[
B_\epsilon \Mac_{k;r} (m) := 
\begin{cases} 
\bot &\text{ if } (m, \Mac_k(m; r)) \in B_\epsilon,\\
\Mac_k(m;r) &\text{ otherwise.}
\end{cases}
\]
where $r$ is sampled uniformly for each oracle call.
\item Outcome: output $1$ if $\Ver_k(m, t) = \acc \wedge (m, t) \in B_\epsilon$; otherwise output $0$.
\end{enumerate}
Security is then defined as before: $\Pi$ is \sBU-secure if
for all adversaries $\algo A$ (and their declared $\epsilon$), the
success probability at winning the above game is negligible. Note that, for the case of canonical MACs, this definition coincides with \expref{Definition}{def:MAC-BU}, just as \EUFCMA and \SUFCMA coincide in this case.

\section{Intuitive security and the meaning of \BU}

In this section, we gather a number of results which build confidence in \BU as a satisfactory definition of unforgeability in our setting. We begin by showing that a wide range of ``intuitively forgeable'' MACs (indeed, all such examples we have examined) are correctly characterized by \BU as insecure.

\subsection{Intuitively forgeable schemes}

As indicated earlier, \BU security rules out any MAC schemes where
an attacker can query a subset of the message space and forge outside
that region. To make this claim precise, we first define the \emph{query support} $\supp(\adver)$
of an oracle algorithm $\adver$. Let $\adver$ be a quantum query
algorithm with oracle access to the quantum oracle $O$ for a classical
function from $n$ to $m$ bits. Without loss of generality $\adver$ proceeds by applying the
sequence of unitaries $\calO U_q \calO U_{q-1}...U_1$ to the initial state
$\ket{0}_{XYZ}$, followed by a POVM $\mathcal E $. Here, $X$ and $Y$
are the input and output registers of the function and $Z$ is the
algorithm's workspace. Let $\ket{\psi_i}$ be the intermediate state of
of $\adver$ after the application of $U_i$. Then $\supp(\adver)$ is
defined to be the set of input strings $x$ such that there exists
a function $f:\bits^n\to \bits^m$ such that
$\bra x_X\ket{\psi_i}\neq 0$ for at least one $i\in\{1,...,q\}$ when
$\calO = \calO_f$. 

\begin{theorem}\label{thm:buinsecure}
  Let $\algo A$ be a QPT such that $\supp(\algo A) \cap R = \emptyset$
  for some $R \neq \emptyset$. Let $\Mac$ be a MAC, and suppose
  $\algo A^{\Mac_k}(1^n)$ outputs a valid pair $(m, \Mac_k(m))$ with
  $m \in R$ with non-negligible probability. Then $\Mac$ is not \BU-secure.
\end{theorem}

To prove \expref{Theorem}{thm:buinsecure}, we will need the following theorem, which controls the change in the output state of an algorithm resulting from applying a blinding to its oracle. Given an oracle algorithm $\algo A$ and two oracles $F$ and $G$, the trace distance between the output of $\algo A$ with oracle $F$ and $\algo A$ with oracle $G$ is denoted by $\delta(\algo A^F(1^n), \algo A^G(1^n))$. Given two functions $F, P : \bits^n \to \bits^m$, we define the function $F \oplus P$ by $(F \oplus P) (x) = F(x) \oplus P(x)$.

\begin{restatable}{theorem}{randomblinding}\label{thm:blinded-algo}
	Let $\algo A$ be a quantum query algorithm making at most $T$ queries, and $F: \bits^n \rightarrow \bits^m$ a function. Let $B \subseteq \{0,1\}^n$ be a subset chosen by independently including each element of $\{0,1\}^n$ with probability $\epsilon$, and $P: \{0,1\}^n \rightarrow \{0,1\}^m$ be any function with support $B$. Then
        \[
          \Exp_B\bigl[ \delta\bigl(\algo A^F(1^n), \algo A^{F \oplus P}(1^n)\bigr)\bigr] \leq 2T\sqrt{\epsilon}.
        \]
\end{restatable}

The proof is a relatively straightforward hybrid argument in the
spirit of the lower bound for Grover search~\cite{BBBV}. We provide
the complete proof in
\expref{Appendix}{sec:simulation-theorem-proof}. We are now ready to
prove \expref{Theorem}{thm:buinsecure}.

\begin{proof}[Proof of {\expref{Theorem}{thm:buinsecure}}] 
  Let $\adver$ be a quantum algorithm with $\supp(\adver)$ for any
  oracle. By our hypothesis,
  \begin{equation*}
    \tilde{p}:={\Pr}_{k,
      (m,t)\leftarrow\adver^{\Mac_k}(1^n)}\left[\Mac_k(m)=t\wedge
      m\notin \supp(\adver)\right]  \geq n^{-c} \, ,    
  \end{equation*}
  for some $c > 0$ and sufficiently large $n$.  Since $\supp(A)$ is a
  fixed set, we can think of sampling a random $B_\veps$ as picking
  $B_0 : = B_\veps \cap\supp(A)$ and
  $B_1: = B_\veps \cap \overline{\supp(A)}$ independently. Let
  ``\textsf{blind}'' denote the random experiment of $\adver$ running
  on $\Mac_k$ blinded by a random $B_\veps$:
  $k,B_\veps,(m,t)\leftarrow\adver^{B_\veps \Mac_k}(1^n)$, which is
  equivalent to $k,B_0, B_1,(m,t)\leftarrow\adver^{B_0 \Mac_k}(1^n)$.
  The probability that $\adver$ wins the BU game is
\begin{align*}
p :=& \Pr_{\text{blind}}[f(m)=t\wedge m\in B_\varepsilon]\geq \Pr_{\text{blind}}[f(m)=t \wedge m\in B']\\
\geq & \Pr_{\text{blind}}[f(m)=t \wedge m\in B' \mid m \notin \supp(A)] \cdot \Pr_{\text{blind}}[m \notin \supp(A)] \\
= &\Pr_{\substack{f,B_0\\(m,t)\gets \adver^{Bf}}}[f(m)=t \wedge m \notin \supp(A)]\cdot \Pr_{\substack{f,B'\\(m,t)\gets \adver^{Bf}}}[m\in B' | m \notin \supp(A)] \\
\geq  &\left(\tilde p - 2T \sqrt \veps \right) \veps\geq \frac{\tilde p^3}{27 T^2} \,.
\end{align*}
Here the second-to-last step follows from \expref{Theorem}{thm:blinded-algo}; in the last step, we chose $\veps = (\tilde p/3T)^2$. We conclude that $\adver$ breaks the \BU security of the MAC.
 \end{proof}

\subsection{Relationship to other definitions}

\subsubsection{Classical \BU is equivalent to \EUFCMA}

In the purely classical setting, our notion is equivalent to EUF-CMA. In the strong unforgeability case, this means \BU with blinding on message-tag pairs, as described in \expref{Section}{sec:strong}.

\begin{proposition}\label{prop:MAC-EUF=BU}
A MAC is EUF-CMA if and only if it is blind-unforgeable against classical adversaries. 
\end{proposition}

\begin{proof}
  Set $F_k = \Mac_k$. Consider an adversary $\mathcal{A}$ which violates EUF-CMA. Such an adversary, given $1^n$ and oracle
  access to $F_k$ (for $k \in_R \{0,1\}^n$), produces a forgery
  $(m, t)$ with non-negligible probability $s(n)$; in particular,
  $|m| \geq n$ and $m$ is not among the messages queried by
  $\mathcal{A}$. This same adversary (when coupled with an appropriate
  $\epsilon$) breaks the system under the blind-forgery
  definition. Specifically, let $p(n)$ be the running time of
  $\mathcal{A}$, in which case $\mathcal{A}$ clearly makes no more
  than $p(n)$ queries, and define $\epsilon(n) = {1}/{p(n)}$. Consider
  now a particular $k \in \{0,1\}^n$ and a particular sequence $r$ of
  random coins for $\mathcal{A}^{F_k}(1^n)$. If this run of
  $\mathcal{A}$ results in a forgery $(m,t)$, observe that with
  probability at least $(1 - \epsilon)^{p(n)} \approx e^{-1}$ in the
  choice of $B_\epsilon$, we have $F_k(q) = B_\epsilon F_k(q)$ for
  every query $q$ made by $\mathcal{A}$. On the other hand,
  $B_\epsilon(m) = \bot$ with (independent) probability $\epsilon$. It
  follows $\phi(n,\epsilon_n)$ is at least
  $\epsilon s(n)/e = \Omega(s(n) / p(n))$.

  On the the other hand, suppose that $(\mathcal{A},\epsilon)$ is an
  adversary that breaks blind-unforgeability. Consider now
  the EUF-CMA adversary $\mathcal{A'}^{F_k}(1^n)$ which simulates the adversary
  $\mathcal{A}^{(\cdot)}(1^n)$ by answering oracle queries according to a
  locally-simulated version of $B_\epsilon F_k$; specifically, the
  adversary $\mathcal{A'}$ proceeds by drawing a subset
  $B_{\epsilon(n)} \subseteq \{0,1\}^*$ as described above and answering
  queries made by $\mathcal{A}$ according to $B_\epsilon F$. Two remarks are in order:
  \begin{itemize}
  \item  When $x \in B_\epsilon$, this query is answered without an
    oracle call to $F(x)$.
  \item $\mathcal{A}'$ can construct the set $B_\epsilon$ ``on the fly,'' by determining, when a particular query $q$ is made by $\mathcal{A}$, whether $q \in B_\epsilon$ and 
``remembering'' this information in case the query is asked again (``lazy sampling'').
  \end{itemize}
  With probability $\phi(n,\epsilon(n))$
  $\mathcal{A}$ produces a forgery on a point which was not queried by
  $\mathcal{A'}$, as desired. It follows that $\mathcal{A}$ produces a
  (conventional) forgery with non-negligible probability when given
  $F_k$ for $k \in_R \{0,1\}^n$. 
  \end{proof}

\subsubsection{\BU implies \GYZ}\label{app:BU=>GYZ}

In this section, we sketch how our new security notion, \BU, implies the one-time security notion put forward by Garg, Yuen and Zhandry \cite{GYZ17}. We do this for unitary adversaries without loss of generality. We expect that the ideas carry over to the fully general case. For the special case we consider here, \GYZ unforgeability can be defined as follows.
\begin{definition}[\GYZ unforgeability]
	Let $\Pi = (\KeyGen, \Mac, \Ver)$ be a MAC. $\Pi$ is called $\varepsilon$-\GYZ-unforgeable, if for attack unitary $V_{MTB}$, there exists a simulator sub-unitary $W_{MTB}$\footnote{A matrix $V$ is sub-unitary if $V^\dagger V\le \mathds 1$.} such that $\bra i_C W\ket j_C=0$ for $i\neq j$ and for all initial states $\ket\psi_{MB}$
	\begin{equation}
		\mathbb E_{r,k}\left\|\left(\Pi_{\Ver, k}\right)_{MT}(V-W)_{MTB}\left(U_{\Mac,r ,k}\right)_{MT}\ket\psi_{MB}\ket 0_T\right\|^2\le \varepsilon,
	\end{equation} 
	where $U_{\Mac,r,k}\ket m_M\ket y_T=\ket m_M\ket{y\oplus\Mac_{r;k}(m)}_T$ and 
	\begin{equation}
	\Pi_{\Ver, k}=\sum_{(m,t)\mathrm{\ valid}}\proj m\otimes \proj{t}
	\end{equation}
	 is the projector onto the subspace of valid message-tag-pairs.
\end{definition}
We are now ready to state the desired theorem. Here, $\delta$-1-\BU denotes the one-time version of \BU that allows for maximal adversarial advantage $\delta$, and similarly for $\delta'$-\GYZ. The theorem is easiest proven using the measured version of $\BU$, \mBU.
\begin{theorem}\label{thm:BU=>GYZ}
  Let $\Pi = (\KeyGen, \Mac, \Ver)$ be unconditionally\footnote{Here, ``unconditionally'' means without any assumption on the computational complexity of the adversary.} $\delta$-1-\mBU-secure. Then it is $16\delta$-\GYZ-unforgeable.
\end{theorem}
\mBU implies \BU, so we immediately obtain the following

\begin{corollary}
	Let $\Pi = (\KeyGen, \Mac, \Ver)$ be $\delta$-1-\BU-secure. Then it is $16\delta$-\GYZ-unforgeable.
\end{corollary}

For the proof of \expref{Theorem}{thm:BU=>GYZ}, we need the following lemma stating that an algorithm's success probability does not degrade too much if some operation that never leaves a computational basis state unchanged is sandwiched between a blinding projector and its complement.

\begin{lemma}\label{lem:change-is-change}
	Let $\Pi^{(\varepsilon)}_A$ be the projector onto the subspace spanned by the computational basis states corresponding to strings in the blinding set $B_\varepsilon$ sampled as for the definition of \BU. Let further $M_{AB}$ be a matrix such that $\bra x_AM\ket x_A=0$ for all $x$. Then for all $\rho\ge 0$
	\begin{equation}
		\mathbb E_{B_\varepsilon}\left[\Pi^{(\varepsilon)}_AM_{AB}\bar\Pi^{(\varepsilon)}_A\rho_{AB} \bar\Pi^{(\varepsilon)}_AM_{AB}^\dagger\Pi^{(\varepsilon)}_A\right]\ge \varepsilon^2(1-\varepsilon)^2M_{AB}\rho_{AB} M_{AB}^\dagger,
	\end{equation}
	where $\bar\Pi^{(\varepsilon)}_A=\mathds 1-\Pi^{(\varepsilon)}_A$. 
\end{lemma}
\begin{proof}
	Let
	\begin{equation}
		M_{AB}=\sum_{x\neq y}\ketbra{x}{y}_{A}\otimes M^{(xy)}_B
	\end{equation}
	and
	\begin{equation}
	\rho_{AB}=\sum_{xy}\ketbra{x}{y}_{A}\otimes \rho^{(xy)}_B.
	\end{equation}
	We calculate
	\begin{align}
		&\mathbb E_{B_\varepsilon}\left[\Pi^{(\varepsilon)}_AM_{AB}\bar\Pi^{(\varepsilon)}_A\rho \bar\Pi^{(\varepsilon)}_AM_{AB}^\dagger\Pi^{(\varepsilon)}_A\right]\\
		&=\sum_{x\neq y,z\neq t}\Pr\left[y,z\in B_\varepsilon, x,t\notin B_\varepsilon\right]\ketbra{x}{t}_A\otimes \left(M^{(xy)}_B\rho^{(yz)}_B\left(M^{(tz)}_B\right)^\dagger \right)\\
		&=\varepsilon^2(1-\varepsilon)^2M_{AB}\rho_{AB} M_{AB}^\dagger+\varepsilon(1-\varepsilon)^3M_{AB}\mathcal M(\rho_{AB} )M_{AB}^\dagger\\
		&+\varepsilon^3(1-\varepsilon)\mathcal M(M_{AB}\rho_{AB}M_{AB}^\dagger )\\
		&+\varepsilon^2(1-\varepsilon)^2\sum_x \proj x_AM_{AB}\proj x_A\rho_{AB}\proj x_AM_{AB}^\dagger \proj x_A.
	\end{align}
	In the last three lines, 
	\begin{equation}
		\mathcal M(X)=\sum_x \proj x X\proj x
	\end{equation} 
	is the computational basis pinching channel,  we used the condition $\bra i_AM\ket i_B=0$ so we can add these terms for free, and  the second, third and fourth term arise because of the increase of the probability $\Pr\left[y,z\in B_\varepsilon, x,t\notin B_\varepsilon\right]$ if $y=z$, $x=t$, or both. But the second, third and fourth term are positive semidefinite, so the claimed operator inequality follows.
\end{proof}

\begin{proof}[Proof of {\expref{Theorem}{thm:BU=>GYZ}}]
  Let 
  \begin{equation}
    V_{MTB} =\sum_{x y}\ketbra{x}{y}_{M}\otimes V^{(xy)}_{TB}
  \end{equation}
  be an attack unitary that breaks $\delta'$-\GYZ, and define the simulator sub-unitary
  \begin{equation}
    W_{MTB} =\sum_{x }\ketbra{x}{x}_{M}\otimes V^{(xx)}_{TB}
  \end{equation}
  to be the diagonal part of $V$. By assumption there exists a
  state $\ket\psi_{MB}$ such that
  \begin{equation}\label{eq:assumpt}
    \mathbb E_k\left\|\left(\Pi_{\Ver, k}\right)_{MT}(V-W)_{MTB}\left(U_{\Mac,r,k}\right)_{MT}\ket\psi_{MB}\ket 0_T\right\|^2> \delta'.
  \end{equation}

  Now we use the initial state $\ket\psi_{MB}\ket 0_T$ and the
  unitary $V$ as attacker $\adver$ for 1-\mBU: $\adver$ queries
  the \mBU-type oracle for $\Mac_k$ on $\ket\psi_{MB}\ket 0_T$,
  applies $V$, measures $MT$ and outputs the result. The success
  probability $p_{\mathrm{succ}}$ of this adversary can be lower
  bounded by the probability that the first measurement turns up
  ``not blinded'' and the adversary is successful. So for a
  fixed key $k$ and randomness $r$, and defining
  $\rho=\proj \phi$ with
  $\ket \phi=\left(U_{\Mac,r,k}\right)_{MT}\ket\psi_{MB}\ket
  0_T$ we get
  \begin{align}
  p_{\mathrm{succ}}(k,r)&\ge \mathbb E_{B_\varepsilon}\left\|\left(\Pi_{\Ver, k}\right)_{MT}\Pi^{(\varepsilon)}_MV_{MTB}\bar\Pi^{(\varepsilon)}_M\left(U_{\Mac,r,k}\right)_{MT}\ket\psi_{MB}\ket 0_T\right\|^2\\
  &= \mathbb E_{B_\varepsilon}\left\|\left(\Pi_{\Ver, k}\right)_{MT}\Pi^{(\varepsilon)}_M(V-W)_{MTB}\bar\Pi^{(\varepsilon)}_M\left(U_{\Mac,r,k}\right)_{MT}\ket\psi_{MB}\ket 0_T\right\|^2\\
  &=\mathbb E_{B_\varepsilon}\Tr\left[\left(\Pi_{\Ver, k}\right)_{MT}\Pi^{(\varepsilon)}_M(V-W)_{MTB}\bar\Pi^{(\varepsilon)}_M\rho\bar\Pi^{(\varepsilon)}_M(V-W)_{MTB}^\dagger\Pi^{(\varepsilon)}_M\right]\\
  &\ge \varepsilon^2(1-\varepsilon)^2\Tr\left[\left(\Pi_{\Ver,k}\right)_{MT}(V-W)_{MTB}\rho(V-W)_{MTB}^\dagger\right]\\
  &=\varepsilon^2(1-\varepsilon)^2\left\|\left(\Pi_{\Ver,k}\right)_{MT}(V-W)_{MTB}\left(U_{\Mac,r,k}\right)_{MT}\ket\psi_{MB}\ket 0_T\right\|^2.
  \end{align}
  
  Here we have used the fact that $\Pi^{(\varepsilon)}_MW_{MTB}\bar\Pi^{(\varepsilon)}_M=0$ for all blinding sets in the second line, and Lemma \ref{lem:change-is-change} in the second-to-last line.
  Taking the expectation over $k$ and $r$ and using Equation \eqref{eq:assumpt}, we arrive at
  \begin{equation}
    p_{\mathrm{succ}}\ge \varepsilon^2(1-\varepsilon)^2\delta'.
  \end{equation}
  Choosing $\epsilon=1/2$ we obtain
  \begin{equation}
    p_{\mathrm{succ}}\ge\frac 1 {16}\delta'.
  \end{equation}
\end{proof}

\subsubsection{\BU implies quadratic \BZ}
\label{sec:bu2qbz}

It is interesting to ask if \BU-security implies \BZ-security, as the
\BZ definition certainly captures a natural family of attacks that one
would like to rule out. We are unable to settle this question
completely, but provide some weaker connection. Specifically, we show
that if a function is BU-secure, then it is \BZ-secure with a weaker
definition of \BZ-security that forbids an adversary from producing
$c k^2$ forgeries from $k$ queries with high probability.

For this purpose, consider a function $M: X \rightarrow Y$ and a
\BZ-type adversary $\mathcal{A}$ which, given oracle access to $M$,
makes some $k$ queries and produces $ck^2$ forgeries (with probability
1); here $c \geq 1$ is a constant we set later in the discussion. We
consider the behavior of this adversary $\mathcal{A}^{B_\epsilon M}$
supplied with an oracle $B_\epsilon M$ blinded at a random set
$B_\epsilon$. We will show that for an appropriate value of $c$ and
$\epsilon$, this adversary produces a family of forgeries which
includes at least one blinded forgery with constant
probability. Finally selecting one of these forgeries at random
produces an adversary that breaks the BU security definition.

Returning to the \BZ-adversary $\mathcal{A}$, we say that a particular
blinding set $B$ is \emph{$\gamma$-evasive} if
\[
\Pr_{\mathcal{A}}[\text{$\mathcal{A}^M$ outputs no elements of $B$}] \geq \gamma\,.
\]
(Note that this event is determined by running $\mathcal{A}$ with the
\emph{unblinded} oracle $M$.)  Observing that
\[
\Pr_{\mathcal{A}, B_\epsilon}[\text{$\mathcal{A}^M$ outputs no
	elements of $B_\epsilon$}] \leq (1 - \epsilon)^{ck^2} \leq e^{-c
	\epsilon k^2}\, . 
\]
We note that (by Markov's inequality),
\[
\Pr_{B_\epsilon}[\text{$B_\epsilon$ is $\gamma$-evasive}] \leq e^{-c\epsilon k}/\gamma\,.
\]
Similarly, we say that a particular blinding set $B$ is
$\gamma$-divergent if
\[
\| D_{\mathcal{A}^M} - D_{\mathcal{A}^{BM}} \|_{\operatorname{t.v.}} \geq \gamma\,,
\]
where $D_{M}$ is the distribution of outputs of $\mathcal{A}^M$ and
$D_{BM}$ is the distribution of outputs of $\mathcal{A}^{BM}$ when $M$
is blinded on set $B$. In light of Theorem~\ref{thm:simulation},
\[
\Exp_{B_\epsilon} \left[ \left\| D_{M} - D_{B_\epsilon M} \right\|_{\operatorname{t.v.}}\right] \leq 2 k \sqrt{\epsilon}
\]
and it follows by Markov's inequality that
\[
\Pr_{B_\epsilon}[\text{$B_\epsilon$ is $\gamma$-divergent}] = \Pr_{B_\epsilon}[\| D_{M} -
D_{BM} \|_{\operatorname{t.v.}} \geq \gamma] \leq
{2k\sqrt{\epsilon}}/{\gamma} \, .
\]

Fixing $\gamma \leq 1/2 - \delta$ for $\delta > 0$, note that if $B$
is neither $\gamma$-evasive nor $\gamma$-divergent then
\[
\Pr_{\mathcal{A}}[\text{$\mathcal{A}^M$ outputs an element of $B$}] \geq 1 - \gamma\,,
\]
(associated with the distribution $D_{M}$), and hence
\[
\Pr_{\mathcal{A}}[\text{$\mathcal{A}^{BM}$ outputs an element of $B$}] \geq 1 - 2\gamma \geq 2\delta\,.
\]
Finally, note that the probability that $B$ is $(1/2-\delta)$-evasive or $(1/2-\delta)$-divergent is no more than
\[
\frac{1}{1/2 - \delta}\underbrace{\left[ e^{-c \epsilon k^2} + 2 k \sqrt{\epsilon}\right]}_{(\dagger)}\,.
\]
Then it is clear that one can choose the constants $\delta$ and $c$,
and the blinding probability $\epsilon = \Theta(1/k^2)$, so that this
quantity is a constant bounded away from one. (For example, set
$\delta = 1/6$. Then, with $\epsilon = 1/(144 k^2)$ the second term of
$(\dagger)$ above is no more than $1/6$; setting $c = 288$ guarantees
the first term is likewise no more than $e^{-2} < 1/6$ and the entire
expression is a constant less than one. One can achieve better constants with more care, but the quadratic dependence on $\epsilon$ in \expref{Theorem}{thm:simulation} dictates the quadratic gap between $k$ and the number of forgeries achieved by this simple method of proof.)

Finally, we create a BU adversary for $M$ by running the \BZ adversary,
blinded as above with $\epsilon = \Theta(1/k^2)$, and selecting one of
the $c k^2$ output values at random.

\section{Blind-unforgeable schemes}
\label{sec:schemes}
\subsection{Random schemes}
\label{ssec:ro}

We now show that suitable random and pseudorandom function families satisfy our notion of unforgeability.

\begin{theorem}\label{thm:random-oracle}
  Let $R: X \rightarrow Y$ be a uniformly random function such that
  $1/{|Y|}$ is negligible in $n$. Then $R$ is a blind-forgery secure
  MAC. 
\end{theorem}
\begin{proof}
  For simplicity, we assume that the function is length-preserving; the proof generalizes easily. Let $\mathcal{A}$ be an efficient quantum adversary. The oracle
  $B_\epsilon R$ supplied to $\mathcal{A}$ during the blind-forgery
  game is determined entirely by $B_\epsilon$ and the restriction of
  $R$ to the complement of $B_\epsilon$. On the other hand, the
  forgery event
  \[ \mathcal{A}^{B_\epsilon F_k}(1^n) = (m,t) \wedge |m| \geq n
    \wedge F_k(m) = t \wedge B_\epsilon F_k(m) = \bot \] depends
  additionally on values of $R$ at points in $B_\epsilon$. To reflect
  this decomposition, given $R$ and $B_\epsilon$ define
  $R_\epsilon: B_\epsilon \rightarrow Y$ to be the restriction
  of $R$ to the set $B_\epsilon$ and note that---conditioned on
  $B_\epsilon R$ and $B_\epsilon$---the random variable $R_\epsilon$
  is drawn uniformly from the space of all (length-preserving)
  functions from $B_\epsilon$ into $Y$. Note, also, that for
  every $n$ the purported forgery
  $(m,t) \leftarrow \mathcal{A}^{B_\epsilon R}(1^n)$ is a (classical)
  random variable depending only on $B_\epsilon R$. In particular,
  conditioned on $B_\epsilon$, $(m,t)$ is independent of
  $R_\epsilon$. It follows that, conditioned on $m \in B_\epsilon$,
  that $t = R_\epsilon(m)$ with probability no more than $1/2^n$ and
  hence $\phi(n,\epsilon) \leq 2^{-n}$, as desired.
 \end{proof}

Next, we show that a $\qPRF$ is a blind-unforgeable MAC.

\begin{corollary}\label{cor:qPRF-MAC}
  Let $m$ and $t$ be $\poly(n)$, and
  $F : \bits^n \times \bits^m \rightarrow \bits^t$ a $\qPRF$. Then $F$
  is a blind-forgery-secure fixed-length MAC (with length $m(n)$).
\end{corollary}
\begin{proof}
For a contradiction, let $\algo A$ be a QPT which wins the blind forgery game for a certain blinding factor $\varepsilon(n)$, with running time $q(n)$ success probability $\delta(n)$. We will use $\algo A$ to build a quantum oracle distinguisher $\algo D$ between the \qPRF $F$ and the perfectly random function family $\mathcal F_m^t$ with the same domain and range.

First, let $k = q(n)$ and let $\mathcal H$ be a family of $(4k+1)$-wise independent functions with domain $\bits^m$ and range $\{0, 1, \dots, 1/\varepsilon(n)\}$. The distinguisher $\algo D$ first samples $h \inrand \mathcal H$. Set $B_h := h^{-1}(0)$. Given its oracle $\algo O_f$, $\algo D$ can implement the function $B_h f$ (quantumly) as follows:
\begin{align*}
\ket{x}\ket{y} 
\mapsto &\ket{x}\ket{y}\ket{H_x}\ket{\delta_{h(x), 0}}
\mapsto \ket{x}\ket{y}\ket{H_x}\ket{\delta_{h(x), 0}}\ket{f(x)}\\
\mapsto &\ket{x}\ket{y \oplus f(x) \cdot (1 - \delta_{h(x), 0})}\ket{H_x}\ket{\delta_{h(x), 0}}\ket{f(x)}\\
\mapsto &\ket{x}\ket{y \oplus f(x) \cdot (1 - \delta_{h(x), 0})}\,.
\end{align*}
Here we used the CCNOT (Toffoli) gate from step 2 to 3 (with one control bit reversed), and uncomputed both $h$ and $f$ in the last step. After sampling $h$, the distinguisher $\algo D$ will execute $\algo A$ with the oracle $B_h f$. If $\algo A$ successfully forges a tag for a message in $B_h$, $\algo A'$ outputs ``pseudorandom''; otherwise ``random.''

Note that the function $B_h f$ is perfectly $\epsilon$-blinded if $h$ is a perfectly random function. Note also that the entire security experiment with $\algo A$ (including the final check to determine if the output forgery is blind) makes at most $2k$ quantum queries and $1$ classical query to $h$, and is thus (by \expref{Theorem}{thm:k-wise}) identically distributed to the perfect-blinding case. 

Finally, by \expref{Theorem}{thm:random-oracle}, the probability that $\algo D$ outputs ``pseudorandom'' when $f \inrand \mathcal F_m^t$ is negligible. By our initial assumption about $\algo A$, the probability that $\algo D$ outputs ``pseudorandom'' becomes $\delta(n)$ when $f \inrand F$. It follows that $\algo D$ distinguishes $F$ from perfectly random.
 \end{proof}

Next, we give a information-theoretically secure $q$-time MACs (\expref{Definition}{def:q-MAC-BU}).

\begin{theorem}
Let $\calH$ be a $(4q+1)$-wise independent function family with range $Y$, such that $1/|Y|$ is a negligible function. Then $\calH$ is a $q$-time BU-secure MAC.
\label{thm:twise}
\end{theorem}
\begin{proof}
Let $(\algo A, \epsilon)$ be an adversary for the $q$-time game $\blindforge^q_{\algo A, h}(n, \epsilon(n))$, where $h$ is drawn from $\calH$. We will use $\algo A$ to construct a distinguisher $\algo D$ between $\calH$ and a random oracle. Given access to an oracle $\calO$, $\algo D$ first runs $\algo A$ with the blinded oracle $B \calO$, where the blinding operation is performed as in the proof of \expref{Corollary}{cor:qPRF-MAC} (i.e., via a $(4q+1)$-wise independent function with domain size $1/\epsilon(n)$). When $\algo A$ is completed, it outputs $(m,\sigma)$. Next, $\algo D$ queries $\calO$ on the message $m$ and outputs 1 if and only if $\calO(m) = \sigma$ and $m \in B$. Let $\gamma_\calO$ be the probability of the output being 1. 

We consider two cases: (i.) $\calO$ is drawn as a random oracle $R$, and (ii.) $\calO$ is drawn from the family $\calH$. By \expref{Theorem}{thm:k-wise}, since $\algo D$ makes only $2q$ quantum queries and one classical query to $\calO$, its output is identical in the two cases. Observe that $\gamma_R$ (respectively, $\gamma_\calH$) is exactly the success probability of $\algo A$ in the blind-forgery game with random oracle $R$ (respectively, $\cal H$). We know from \expref{Theorem}{thm:random-oracle} that $\gamma_R$ is negligible; it follows that $\gamma_\calH$ is as well. 
 \end{proof}

Several domain-extension schemes, including
NMAC (a.k.a.\ encrypted cascade), HMAC, and AMAC, can transform a
fixed-length \qPRF~to a \qPRF~that takes variable-length
inputs~\cite{SY17} . As a corollary, starting from a \qPRF,
we also obtain a number of quantum blind-unforgeable variable-length MACs.

\subsection{Lamport one-time signatures}\label{sec:Lamport}

The Lamport signature scheme \cite{Lam79} is a EUF-1-CMA-secure signature scheme, specified as follows.%
\begin{construction}[Lamport signature scheme, \cite{Lam79}]\label{con:lamport}
	For the Lamport signature scheme using a hash function family $h: \{0,1\}^n\times\bits^n\to\bits^n$, the algorithms $\KeyGen, \Sign$ and $\Ver$ are specified as follows. $\KeyGen$, on input $1^n$, outputs a pair $(\pk,\sk)$ with 
	\begin{align}
	\sk &=(s_i^j)_{i\in\{1,...,n\}, j=0,1},\text{ with }s_i^j\in_R\bits^n\text{, and}\\
	\pk &=\left(k,\left(p_i^j\right)_{i\in\{1,...,n\}, j=0,1}\right), \text{ with }k\in\bits^n \text{ and } p_i^j=h_k\left(s_i^j\right).
	\end{align}
	The signing algorithm is defined by $\Sign_{\sk}(x)=(s_i^{x_i})_{i\in\{1,...,n\}}$
	where $x_i$, $i=1,...,n$ are the bits of $x$. The verification procedure checks  the signature's consistency with the public key, i.e.,	$\Ver_{\pk}(x,s)=0$ if $p_i^{x_i}=h_k(s_i)$ and  $\Ver_{\pk}(x,s)=0$ otherwise.
\end{construction}

We now show that the Lamport scheme is 1-\BU secure in the quantum random oracle model.

\begin{theorem}\label{thm:lamport}
  \expref{Construction}{con:lamport} is 1-\BU secure if $h$ is modeled as
  a quantum-accessible random oracle.
\end{theorem}

\begin{proof}
	We implement the random oracle $h$ as a superposition oracle with register $F$. In the 1-\blindforge experiment we execute the sampling part of the key generation by preparing a superposition as well. More precisely, we can just prepare $2n$ $n$-qubit registers $S_i^j$ in a uniform superposition, with the intention of measuring them to sample $s_i^j$ in mind. We are talking about a classical one-time signature scheme, and all computation that uses the secret key is done by an honest party, and is therefore classical. It follows that the measurement that samples $s_i^j$ commutes with all other operations which are implemented as quantum-controlled operations controlled on the secret key registers, i.e., we can postpone it to the very end of the 1-\blindforge experiment, just like the measurement that samples an actual random oracle using a superposition oracle. The joint state $\ket{\psi_0}$ with oracle register F and secret key register $\skreg =(S_i^j)_{i\in\{1,...,n\}, j=0,1}$ is now in a uniform superposition, i.e.,
	\begin{equation}
	\ket{\psi_0}_{\skreg F}=\ket{\phi_0}^{\otimes 2n}_{\skreg }\otimes \ket{\phi_0}^{\otimes 2^n}_F.
	\end{equation}
	To subsequently generate the public key, the superposition oracle for $h$ is queried on each of the $S_i^j$ with an empty outrput register $P_i^j$, producing the state $\ket{\psi_1}_{\skreg\pkreg F}$ equal to
	\[
	2^{-2n^2}\sum_{\substack{s_i^j\in\bits^n\\p_i^j\in\bits^n\\i\in\{1,...,n\}, j=0,1}}\left(\bigotimes_{\substack{i\in\{1,...,n\}\\j=0,1}}\ket{s_i^j}_{S_i^j}\right)\otimes \left(\bigotimes_{\substack{i\in\{1,...,n\}\\j=0,1}}\ket{p_i^j}_{P_i^j}\right)\otimes\ket{f_{\sk,\pk}}_F,
	\]
	where $\ket{f_{\sk,\pk}}_F$ is the superposition oracle state where $F_{s_i^j}$ is in state $\ket{p_i^j}$ and all other registers are still in state $\ket{\phi_0}$.
	Then the registers $P_i^j$ are measured to produce an actual, classical, public key that can be handed to the adversary. Note that there is no hash function key $k$ now, as it has been replaced by the random oracle. Treating the public key as classical information from now on and removing the registers $\pkreg$, the state takes the form
	\begin{align}
	\ket{\psi_2(\pk)}_{\skreg F}=&2^{-n^2}\sum_{\substack{s_i^j\in\bits^n\\i\in\{1,...,n\}, j=0,1}}\left(\bigotimes_{\substack{i\in\{1,...,n\}\\j=0,1}}\ket{s_i^j}_{S_i^j}\right)\otimes\ket{f_{\sk,\pk}}_F.
	\end{align}
	Now the interactive phase of the 1-\blindforge experiment can begin, and we provide both the random oracle $h$ and the signing oracle (that can be called exactly once) as superposition oracles using the joint oracle state $\ket{\psi_2(\pk)}$ above. The random oracle answers queries as described in \expref{Section}{sec:superoracles}. The signing oracle, when queried with registers $XZ$ with $Z=Z_1...Z_n$, applies $\CNOT^{\otimes n}_{S_i^{x_i}:Z_i}$, $i=1,...,n$ controlled on $X$ being in the state $x\notin B_\varepsilon$.
	
	Now suppose \adver, after making at most one query to $\Sign$ and an arbitrary polynomial number of queries to $h$, outputs a candidate message signature pair $(x^0, z^0)$ with $z^0=z^0_1\|\cdots \|z^0_n$. If $x^0\notin B_\varepsilon$, \adver has lost. Suppose therefore that $x^0\in B_\varepsilon$. We will now make a measurement on the oracle register to find an index $i$ such that $S_i^{x^0_i}$ has not been queried. To this end we first need to decorrelate $\skreg$ and $F$. This is easily done, as the success test only needs computational basis measurement results from the register \skreg, allowing us to perform any controlled operation on $F$ controlled on $\skreg$. Therefore we can apply the operation $\bigoplus p_i^j$ followed by $H^{\otimes n}$ to the register $F_{s_i^j}$ controlled on $S_i^j$ being in state $\ket{s_i^j}$, for all $i=1,...,n$ and $j=0,1$. For an adversary that does not make any queries to $h$, this has the effect that all $F$-registers are in state $\ket{\phi_0}$ again now.
	
	We can equivalently perform this restoring procedure before the adversary starts interaction, and answer the adversary's $h$-queries as follows. Controlled on the adversary's input being equal to one of the parts $s_i^j$ of the secret key, answer with the corresponding public key, otherwise use the superposition oracle for $h$.
	
	For any fixed secret key register $S_i^j$, the unitary that is applied upon an $h$-query can hence be written as
	\begin{align}
	U_h'&=U_\bot+\sum_{x\in\bits^n}(U_x-U_\bot)\proj x_X\proj x_{S_i^j}\\
	&=U_\bot+\sum_{x\in\bits^n}\proj x_X\proj x_{S_i^j}(U_x-U_\bot),
	\end{align}
	where $U_\bot$ acts trivially on $S_i^j$ and the second equality follows because the unitaries $U_\bot$ and $U_x$ are controlled unitaries with $X$ and $S_i^j$ part of the control register. Using the above equation we derive a bound on the operator norm of the commutator of this unitary and the projector onto $\ket{\phi_0}$,
	\begin{align*}
	\left\|\left[U_h',\proj{\phi_0}_{S_i^j}\right]\right\|_\infty =&2^{-n/2}\left\|\sum_{x\in\bits^n}\!\!\left((U_x-U_\bot)\proj x_X\ketbra{x}{\phi_0}_{S_i^j}-\proj x_X\ketbra{\phi_0}{x}_{S_i^j}(U_x-U_\bot)\right)\right\|_\infty\\
	=&2^{-n/2}\!\!\max_{x\in\bits^n}\left\|\left((U_x-U_\bot)\proj x_X\ketbra{x}{\phi_0}_{S_i^j}-\proj x_X\ketbra{\phi_0}{x}_{S_i^j}(U_x-U_\bot)\right)\right\|_\infty\\
	\le& 4\cdot 2^{-n/2},
	\end{align*}
	where the second equality follows again because $U_\bot$ and $U_x$ are controlled unitaries with $X$ and $S_i^j$ part of the control register.
	
	It follows that a query to $h$ does not decrease the number of registers $S_i^j$ that are in state $\ket{\phi_0}$, except with negligible amplitude.
	
	As we assume that $x^0$ is blinded, we have that for any message $x\notin B_\varepsilon$, there exists an $i\in\{1,...,n\}$ such that $x_i\neq x^0_i$. But \adver interacts with a blinded signing oracle, i.e., controlled on his input being not blinded, it is forwarded to the signing oracle, otherwise $\bot$ is XORed into his output register. Therefore only non-blinded queries have been forwarded to the actual signing oracle, so the final state is a superposition of states in which the register $\skreg$ has at least $n$ subregisters $S_i^j$ are in state $\ket{\phi_0}$, and at least one of them is such that $x^0_i=j$. We can therefore apply an $n$-outcome measurement to the oracle register to obtain an index $i_0$ such that $S_{i_0}^{x^0_{i_0}}$ is in state $\ket{\phi_0}$.  By \expref{Lemma}{lem:add-measure}, this implies that \adver's forgery is independent of $s_{i_0}$, so \adver's probability of succeeding in \blindforge is negligible.
 \end{proof}

A simple proof of the \BZ-security of a random function can be given using a similar idea, see \expref{Theorem}{thm:Rand-BZ} in the appendix.

\subsection{Hash-and-MAC}
\label{ssec:hash-and-mac}

To authenticate messages of arbitrary length with a fixed-length MAC, it is common practice to first compress a long message by a
\emph{collision-resistant} hash functon and then apply the MAC. This
is known as Hash-and-MAC. However, when it comes to
\BU-security, collision-resistance may not be sufficient. We therefore
propose a new notion, \bfilter, generalizing collision-resistance in
the quantum setting, and show that it is sufficient for Hash-and-MAC with
\BU security. Recall that, given a subset $B$ of a set $X$, $\chi_B : X \to \bits$
denotes the characteristic function of $B$.

\begin{definition}[\bfilter]\label{def:bphash}
Let $\mathcal H : X \to Y$ be an efficiently computable function family. Define the following distributions on subsets of $X$:
\begin{enumerate}
\item $\mathcal B_\epsilon$ : generate $B_\epsilon \subseteq X$ by placing $x \in B_\epsilon$ independently with probability $\epsilon$. Output $B_\epsilon$.
\item $\mathcal B_\epsilon^\mathcal H$ : generate $C_\epsilon \subseteq Y$ by placing $y \in C_\epsilon$ independently with probability $\epsilon$. Sample $h \in \mathcal H$ and define $B_\epsilon^h := \{x \in X : h(x) \in C_\epsilon\}$. Output $B_\epsilon^h$.
\end{enumerate}
We say that $\mathcal H$ is a \bfilter if for all adversaries $(\algo A, \epsilon)$,
$$
\Bigl|\Pr_{B \from \mathcal B_\epsilon} \left[ \algo A^{\chi_B}(1^n) = 1\right]
- \Pr_{B \from \mathcal B_\epsilon^\mathcal H} \left[ \algo A^{\chi_B}(1^n) = 1\right]\Bigr|
\leq \negl\,.
$$
\end{definition}

The motivation for the name \bfilter is simply that selecting
$\mathcal B_\epsilon$ can be viewed as a Bernoulli process taking
place on the set $X$, while $\mathcal B_\epsilon^h$ can be viewed as
the pullback (along $h$) of a Bernoulli process taking place on $Y$.

We show that the standard, so-called ``Hash-and-MAC'' construction
will work w.r.t. to \BU security, if we instantiate the hash funtion
with a \bfilter. Recall that, given a MAC $\Pi = (\Mac_k, \Ver_k)$
with message set $X$ and a function $h: Z \to X$, there is a MAC
$\Pi^h := (\Mac_k^h, \Ver_k^h)$ with message set $Z$ defined by
$\Mac_k^h = \Mac_k \circ h$ and $\Ver_k^h(m, t) = \Ver_k(h(m), t)$.

\begin{theorem}[Hash-and-MAC with \bfilter]
  Let $\Pi = (\Mac_k, \Ver_k)$ be a \BU-secure MAC with
  $\Mac_k : X \to Y$, and let $h: Z \to X$ a \bfilter. Then $\Pi^h$ is
  a \BU-secure MAC.
\end{theorem}

\begin{proof}
Let $\algo A$ be an adversary against $\Pi^h$. We build an adversary $\algo A_0$ against $\Pi$ which (given oracle $f: X \to Y$) runs $\algo A$ and answers its queries with $f \circ h$, i.e., $\ket{m}\ket{t} \mapsto \ket{m}\ket{t \oplus f(h(m))}$. This can be implemented by first computing $h$ into an extra register, then invoking the oracle, and then uncomputing $h$. When $\algo A$ produces its final output $(m, t)$, $\algo A_0$ outputs $(h(m), t)$ and terminates. We claim that 
\[
\left|\Pr[\blindforge_{\algo A, \Pi^h}(n, \epsilon) = 1] - \Pr[\blindforge_{\algo A_0, \Pi}(n, \epsilon) = 1]\right| \leq \negl\,.
\]
Since the right-hand-side of the difference above is negligible by \BU-security of $\Pi$, establishing the claim will finish the proof.

We prove the claim by showing that the difference can be viewed as the success probability of a distinguisher $\algo D$ against the Bernoulli-preserving property of $h$. The distinguisher $\algo D$ receives an oracle for $\chi_B$ (where $B \subseteq Z$ is sampled according to either $\mathcal B_\epsilon$ or $\mathcal B^h_\epsilon$) and proceeds as follows:
\begin{enumerate}
\item generate a key $k$ for $\Pi$;
\item run $\algo A$, answering its oracle queries with
\[
\ket{m}\ket{t} 
\mapsto \ket{m}\ket{t}\ket{\chi_B(m)}\ket{\Mac_k(h(m))}\\
\mapsto \ket{m}\ket{t \oplus \chi_B(m) \cdot \Mac_k(h(m))}\ket{\chi_B(m)}
\]
where we invoked the oracle in the first step and CCNOT in the second;
\item when $\algo A$ outputs $(m, t)$, compute $b = \Ver^h_k(m, t) = \Ver_k(h(m), t)$. Query the oracle to compute $b' = \chi_B(m)$, and output $b \wedge b'$.
\end{enumerate}
It now remains to check that (i.)
if $B$ was sampled according to $\mathcal B_\epsilon$ (i.e., uniform blinding), then $\algo D$ is simulating the game $\blindforge_{\algo A, \Pi^h}(n, \epsilon)$, and (ii.) If $B$ was sampled according to $\mathcal B_\epsilon^h$ (i.e., hash-blinding), then $\algo D$ is simulating the game $\blindforge_{\algo A_0, \Pi}(n, \epsilon)$. Fact (i.) follows directly from the definition\footnote{Note that we have again used the convention that the blinding symbol $\bot$ is the string $0\dots01$; in our case, the final bit corresponds to the register containing $\chi_B(m)$. If one chooses a different convention, it may be necessary to adjust $\algo D$ to uncompute that register with an extra call to the oracle.} of the $\blindforge$ game.  To see fact (ii.), observe that the $\blindforge$ game against $\algo A_0$ samples a uniform blinding set $C_\epsilon \subseteq X$ and executes algorithm $\algo A$ with oracle 
$$
m \longmapsto \chi_{C_\epsilon}(h(m)) \cdot \Mac_k(h(m)) = \chi_{B_\epsilon^h}(m) \cdot \Mac_k(h(m))\,,
$$
precisely as in the execution of $\algo A$ by $\algo D$.
\end{proof}

In the next section, we provide a number of additional results about
\bfilter functions.

\section{Properties of \bfilter functions}
We explore the notion of \bfilter functions. These results can be
summarized as follows.
\begin{itemize}
\item If $H$ is a random oracle or a \qPRF, then it is a \bfilter.
\item If $H$ is $4q$-wise independent, then it is a \bfilter against $q$-query adversaries.
\item Under the \LWE assumption, there is a (public-key) family of \bfilter functions.
\item If we only allow classical oracle access, then the Bernoulli-preserving property is  equivalent to standard collision-resistance.
\item \bfilter functions are \emph{collapsing} (another quantum generalization of collision-resistance proposed in~\cite{Unruh16a}).
\end{itemize}

First, we show that random and pseudorandom functions are Bernoulli-preserving, and that this property is equivalent to collision-resistance against classical queries.
\begin{lemma}
Let $H : X \to Y$ be a function such that $1/|Y|$ is negligible. Then:
\begin{enumerate}
\item If $H$ is a random oracle or a \qPRF, then it is a \bfilter.
\item If $H$ is $4q$-wise independent, then it is a \bfilter against $q$-query adversaries.
\end{enumerate}
\end{lemma}
\begin{proof}
The claim for random oracles is obvious: by statistical collision-resistance, uniform blinding is statistically indistinguishable from hash-blinding. The remaining claims follow from the observation that one can simulate one quantum query to $\chi_{B_\epsilon^h}$ using two quantum queries to $h$ (see, e.g., the proof of \expref{Corollary}{cor:qPRF-MAC}).
 \end{proof}

\begin{theorem}
A function $h: \bits^*\to \bits^n$ is Bernoulli-preserving against classical-query adversaries if and only if it is collision-resistant.
\end{theorem}
\begin{proof}
  First, the \bfilter property implies collision-resistance: testing whether two colliding inputs are either (i.) both not blinded or both blinded, or (ii.) exactly one of them is blinded, yields always outcome (i.) when dealing with a hash-blinded oracle and a uniformly random outcome for a blinded oracle and $\varepsilon=1/2$. On the other hand, consider an adversary $\adver$ that has inverse polynomial distinguishing advantage between blinding and hash-blinding, and let $x_1,...,x_q$ be it's queries. Assume for contradiction that with overwhelming probability $h(x_i)\neq h(x_j)$ for all $x_i\neq x_j$. Then with that same overwhelming probability the blinded and hash blinded oracles are both blinded independently with probability $\varepsilon$ on each $x_i$ and are hence statistically indistinguishable, a contradiction. It follows that with non-negligible probability there exist two queries $x_i\neq x_j$ such that $h(x_i)=h(x_j)$, i.e., $\adver$ has found a collision.
\end{proof}

\subsection{A \bfilter from LWE}\label{sec:lwe}
We have observed that any \qPRF is a \bfilter function,
which can be constructed from various quantum-safe computational
assumpiton (e.g., LWE). Nonetheless, \qPRF typically does not give
short digest, which would result in long tags, and it requires a secret key.\footnote{In practice, it
  is probably more convenient (and more reliable) to instantiate a
  \qPRF from block ciphers, which may not be ideal for message
  authentication.} 

Here we point out an alternative construction of a public \bfilter
function based on the quantum security of LWE. In fact, we show that
the collapsing hash function by Unruh~\cite{UnruhAC16} is also
\bfilter. This constructions relies on a lossy function family
$F: X\to Y$ and a universal hash function
$ G = \{g_k: Y \to Z \}_{k\in \mathcal{K}}$. A lossy function family
admits two types of keys: a lossy key
$s\gets \mathcal{D}_{\mathrm{los}}$ and an injective key
$s\gets \mathcal{D}_{inj}$, which are computationally
indistinguishable. $F_{s}: X \to Y$ under a lossy key $s$ is
compressing, i.e., $|\operatorname{im}(F_s)| \ll |Y|$; whereas under
an injective key $s$, $F_s$ is injective. We refer a formal definition
to~\cite[Definition 2]{UnruhAC16}, and an explicit construction based
on LWE to~\cite{PW08}. We will also use exist efficient constructions
for universal hash families~\cite{Vad12}. Then one constructs a hash
funciton family $H = \{h_{s,k}\}$ by $ h_{s,k} := g_k\circ F_s$ with
public parameters generated by
$s\gets \mathcal{D}_{\mathrm{los}}, k \gets \mathcal{K}$.

The proof of Bernoulli-preserving for this hash function is similar to Unruh's proof that $H$ is collapsing. We begin with a lemma.

\begin{lemma} Any injective function $f$ is \bfilter. Given any
  \bfilter $f : X\to Y$ and $g: Y\to Z$ that is \bfilter on $\operatorname{im}(f)$,
  then $h = g\circ f$ is also \bfilter.
  \label{lemma:bfilter}
\end{lemma}
\begin{proof} The first part follows by observing that a
  $\veps$-random subset in the codomain corresponds exactly to a
  $\veps$-random subset in the domain under inverse of the function.
  Let $\calO \approx \calO'$ denote that two oracles $\calO$ and
  $\calO$ are indistinguishable by any quantum poly-time algorithm.
  For the second part, we need to show that
  $\chi_{C: C \gets_\veps X} \approx \chi_{C: C=h^{-1}(C_Z),
    C_Z\gets_\veps Z}$, where $\gets_\veps $ indicates sampling a
  random subset of fraction $\veps$. Since $f$ is \bfilter, we have
  that
  \begin{equation*}
    \chi_{C: C \gets_\veps X} \approx \chi_{C: C=f^{-1}(C_Y),
      C_Y\gets_\veps Y} \equiv  \chi_{C: C=f^{-1}(C'_Y), C'_Y
      \gets_\veps \operatorname{im}(f)}\, .
  \end{equation*}
  The second equivalence holds by observing that for any
  $C_Y\subseteq Y$, $f^{-1}(C_Y) = f^{-1}(C_Y\cap \operatorname{im}(f))$. Then
  because $g$ is Bernoulli-preserving on $\operatorname{im}(f)$,
  \begin{equation*}
    \chi_{C'_Y: C'_Y  \gets_\veps \operatorname{im}(f)} \approx     \chi_{C'_Y: C'_Y =
      g^{-1}(C_Z), C_Z \gets_\veps Z}  \, .
  \end{equation*}
  Therefore, we conclude that
  \begin{equation*}
    \chi_{C: C \gets_\veps X} \approx \chi_{C:
      C=f^{-1}(g^{-1}(C_Z)), C_z\gets_\veps Z} = \chi_{C:
      C= h^{-1}(C_Z), C_z\gets_\veps Z} \, . \qedhere\qquad
  \end{equation*}
\end{proof}

\begin{theorem} $H$ is \bfilter if LWE holds against any efficient
  quantum distinguisher.
  \label{thm:lwebp}
\end{theorem}

\begin{proof} We proceed in three steps (with the help of
  \expref{Lemma}{lemma:bfilter} above):

  \begin{enumerate}[label=\arabic*)]
  \item Since $F_s$ is injective under an injective key, it is clearly
    \bfilter. As a result, $F_s, s\gets \mathcal{D}_{\mathrm{los}}$ must be
    \bfilter too, because a lossy key is indistinguishable from an
    injective key by definition.
  \item Then $g_k$ is chosen properly so that it is injective when
    restricted to $\operatorname{im}(F_s)$ of lossy key $s$. Therefore $g_k$ is
    \bfilter too.
  \item Finally, $H_{k,s}$ is \bfilter by the composition of \bfilter
    functions $g_k$ and $F_s$. 
  \end{enumerate}
\end{proof}

\subsection{Relationship to collapsing}

Here we relate \bfilter to the collapsing property, which is another
quantum generalization of classical collision-resistance. We first
describe the collapsing property (slightly adapting Unruh's original
definition~\cite{Unruh16a}) as follows. Let $h: X \to Y$ be a hash
function, and let $\mathcal S_X$ and $\mathcal S_{XY}$ be the set of
quantum states (i.e., density operators) on registers corresponding to
the sets $X$ and $X \times Y$, respectively. We define two channels
from $\mathcal S_X$ to $\mathcal S_{XY}$. First, $\calO_h$ receives
$X$, prepares $\ket{0}$ on $Y$, applies
$\ket{x}\ket{y} \mapsto \ket{x}\ket{y \oplus h(x)}$, and then measures
$Y$ fully in the computational basis. Second, $\calO'_h$ first applies
$\calO_h$ and then also measures $X$ fully in the computational basis.
\begin{align*}
  \calO_h: \quad & \ket{x}_{X} \overset{h}{\longmapsto} \ket{x, h(x)}_{X,Y} \overset{\text{~~measure } Y\text{~~~}}{\longmapsto}
                   (\rho_{X}^y, y)\, , \\
  \calO'_h: \quad & \ket{x}_{X} \overset{h}{\longmapsto}  \ket{x, h(x)}_{X,Y}
                    \overset{\text{measure } X\& Y}{\longmapsto} (x, y)\, .
\end{align*}
If the input is a pure state on $X$, then the output is either a superposition over a fiber $h^{-1}(s) \times \{s\}$ of $h$ (for $\calO_h$) or a classical pair $(x, h(x))$ (for $\calO_h'$) . 

\begin{definition}[Collapsing] A hash function $h$ is collapsing if for any single-query QPT $\adver$, it holds that
$    
\bigl|\Pr[\adver^{\calO_h}(1^n) =1 ]- \Pr[\adver^{\calO_h'}(1^n) = 1]
    \bigr| \leq \negl \, .
$
\label{def:collapsingalt}  
\end{definition}

To prove that \bfilter implies collapsing, we need a technical fact. Recall that any subset $S \subseteq \bits^n$ is associated with a two-outcome projective measurement $\{\Pi_S, \one - \Pi_S\}$ on $n$ qubits defined by $\Pi_S = \sum_{x \in S} \outerprod{x}{x}$. We will write $\Xi_S$ for the channel (on $n$ qubits) which applies this measurement. 

\begin{lemma}\label{lemma:partial-measure}
  Let $S_1, S_2, \dots, S_{cn}$ be subsets of $\bits^n$, each of size
  $2^{n-1}$, chosen independently and uniformly at
  random. 
  Let $\Xi_{S_j}$ denote the two-outcome measurement defined by $S_j$,
  and denote their composition
  $\tilde \Xi : = \Xi_{S_{cn}} \circ \Xi_{S_{cn-1}} \circ \cdots \circ
  \Xi_{S_1}$. Let $\Xi$ denote the full measurement in the
  computational basis. Then\\
$
\Pr\bigl[\tilde \Xi = \Xi \bigr] \geq 1 - 2^{-\veps n} \, ,
$
whenever $c \geq 2 + \veps$ with $\veps > 0$,
\end{lemma}

\begin{proof} We give a combinatorial proof. Consider an arbitrary
  mixed state of density matrix
  $\rho = (\rho_{x,y})_{x,y \in \bits^n}$, the full measurement $\Xi$
  on $\rho$ gives

  \begin{equation*}
    \Xi(\rho) = \sum_{x \in \bits^n} \ketbra{x}{x} \rho \ketbra{x}{x}
    = \sum_{x\in \bits^n} \rho_{x,x} \ketbra{x}{x}\,.
  \end{equation*}
  Given a set $S \subseteq \bits^n$, the projective measurement
  $\Xi_S$ on $\rho$ operates as
  \begin{align*}
    \Xi_S (\rho) & = \sum_{x,y\in S}\ketbra{x}{x} \rho \ketbra{y}{y}
                   +  \sum_{x,y\notin S}\ketbra{x}{x} \rho \ketbra{y}{y} \\
                 & =
                   \sum_{x,y \in S} \rho_{x,y} \ketbra{x}{y} +
                   \sum_{x,y \notin S} \rho_{x,y} \ketbra{x}{y} \, .
  \end{align*}
  Namely, $\Xi_S$ will zero-out the entries $\rho_{x,y}$ in $\rho$,
  where $(x\in S, y\notin S)$ or $(x\notin S, y \in S)$. It is easy to
  verify that the same effect occurs when $\Xi$ and $\Xi_S$ are
  applied to a subsystem of a bipartite state.

  Now, for any $c = 2+\veps$ with $\veps>0$, consider sampling
  $S_1,S_2, \ldots, S_{cn}$ independently at random, each of size
  $2^{n-1}$, and define a few random events:
  \begin{align*}
    & E^i_{x,y}:  x\in S_i \wedge y \in S_i, \text{ or } x\notin S_i
      \wedge y \notin S_i \, ; \\
    & E_{x,y}: \forall i \in \{1, \ldots, cn\} \text{ s.t. } E^i_{x,y}
      \, ; \\
    & \bad: \exists x,y \in \bits^n, x\neq y \text{ s.t. } E_{x,y} \, .
  \end{align*}
  Observe that if $\bad$ does not occur, it implies that for any
  $x\neq y$, the off-diagonal entry $\rho_{x,y}$ is eliminated by one
  of $\Xi_{S_i}$, and as a result
  $\tilde \Xi = \Xi_{S_{cn}} \circ \ldots \circ \Xi_{S_1}$ will be
  identical to $\Xi$.

  Fix a pair $(x,y)$ with $x\neq y$, clearly $\Pr[E^i_{x,y}] =
  1/2$. Since each $S_i$ is chosen independently,
  \begin{equation*}
    \Pr[E_{x,y}] = \Pi_i \Pr[E^{i}_{x,y}] = 1/{2^{cn}} \, .
  \end{equation*}
  By the union bound,
  \begin{equation*}
    \Pr[\bad] \leq \binom{2^n}{2} \cdot \Pr[E_{x,y}] \leq
    2^{2n}/{2^{cn}} = 2^{-\veps n} \, .
  \end{equation*}
  Therefore we conclude that
  \begin{equation*}
    \Pr[\tilde \Xi = \Xi] \geq \Pr[\tilde \Xi = \Xi \mid \overline{\bad}]
    \cdot \Pr[\overline{\bad}] \geq 1 - 2^{-\veps n}\, . \qedhere
  \end{equation*}
\end{proof}

We remark that to efficiently implement each $\Xi_S$ with a random
subset $S$, we can sample $h_i: [M] \to [N]$ from a
pairwise-independent hash family (sampling an independent $h_i$ for
each $i$), and then define $x\in S$ iff $h(x) \leq N/2$. For any
input state $\sum_{x,z}\alpha_{x,z} \ket{x,z}$, we can compute
\[ 
\sum_{x,z}\alpha_{x,z} \ket{x,z} 
\mapsto \sum_{x,z} \ket{x,z} \ket{b(x)}, \quad \text{ where } b(x) : = h(x) \overset{?}{\leq}
  N/2\,,
\] 
and then measure $\ket{b(x)}$. Pairwise independence is sufficient by \expref{Theorem}{thm:k-wise} because only one quantum query is made.

\begin{theorem}
If $h: X \to Y$ is Bernoulli-preserving, then it is collapsing. 
\label{thm:bf2collapsing}  
\end{theorem}
\begin{proof}
Let $\algo A$ be an adversary with inverse-polynomial distinguishing power in the collapsing game. Choose $n$ such that $X = \bits^n$. We define $k = cn$ hybrid oracles $H_0, H_1, \dots, H_k$, where hybrid $H_j$ is a channel from $\mathcal S_X$ to $\mathcal S_{XY}$ which acts as follows:
(1.) adjoin $\ket{0}_Y$ and apply the unitary $\ket{x}_X\ket{y}_Y \mapsto \ket{x}_X\ket{y \oplus h(x)}_Y$; (2.) measure the $Y$ register in the computational basis; (3.) repeat $j$ times: (i.) select a uniformly random subset $S \subseteq X$ of size $2^{n-1}$; (ii.) apply the two-outcome measurement $\Xi_S$ to the $X$ register; (4.) output registers $X$ and $Y$.

Clearly, $H_0$ is identical to the $\calO_h$ channel in the collapsing game. By \expref{Lemma}{lemma:partial-measure}, $H_k$ is indistinguishable from the $\calO_h'$. By our initial assumption and the triangle inequality, there exists a $j$ such that
\begin{equation}\label{eq:collapse-hybrid}
\left|\Pr[\algo A^{H_j}(1^n) = 1] - \Pr[\algo A^{H_{j+1}}(1^n) = 1]\right| \geq 1/\poly\,.
\end{equation}

We now build a distinguisher $\algo D$ against the Bernoulli-preserving property (with $\epsilon = 1/2$) of $h$. It proceeds as follows:
(1.) run $\algo A(1^n)$ and place its query state in register $X$; (2.) simulate oracle $H_j$ on $XY$ (use 2-wise independent hash to select sets $S$); (3.) prepare an extra qubit in the $\ket{0}$ state in register $W$, and invoke the oracle for $\chi_B$ on registers $X$ and $W$; (4.) measure and discard register $W$; (5.) return $XY$ to $\algo A$, and output what it outputs.

We now analyze $\algo D$. After the first two steps of $H_j$ (compute $h$, measure output register) the state of $\algo A$ (running as a subroutine of $\algo D$) is given by
$$
\sum_z \sum_{x \in h^{-1}(s)} \alpha_{xz} \ket{x}_X \ket{s}_Y \ket{z}_Z\,.
$$
Here $Z$ is a side information register private to $\algo A$. Applying
the $j$ measurements (third step of $H_j$)  results in
a state of the form
$
\sum_z \sum_{x \in M} \beta_{xz} \ket{x} \ket{ s} \ket{z}\,,
$
where $M$ is a subset of $h^{-1}(s)$. Applying the oracle for $\chi_B$ into an extra register now yields
$$
\sum_z \sum_{x \in M} \beta_{xz} \ket{x} \ket{ s} \ket{z}
\ket{\chi_B(x)}_W\,.
$$
Now consider the two cases of the Bernoulli-preserving game. 

First, in the ``hash-blinded'' case, $B = h^{-1}(C)$ for some set $C \subseteq Y$. This implies that $\chi_B(x) = \chi_C(h(x)) = \chi_C(s)$ for all $x \in M$. It follows that $W$ simply contains the classical bit $\chi_C(s)$; computing this bit, measuring it, and discarding it will thus have no effect. The state returned to $\algo A$ will then be identical to the output of the oracle $H_j$. Second, in the ``uniform blinding'' case, $B$ is a random subset of $X$ of size $2^{n-1}$, selected uniformly and independently of everything else in the algorithm thus far. Computing the characteristic function of $B$ into an extra qubit and then measuring and discarding that qubit implements the channel $\Xi_B$, i.e., the measurement $\{\Pi_B, \one - \Pi_B\}$. It follows that the state returned to $\algo A$ will be identical to the output of oracle $H_{j+1}$.

By \eqref{eq:collapse-hybrid}, it now follows that $\algo D$ is a
successful distinguisher in the \bfilter game for $h$. Hence $h$ is
not a \bfilter.
\end{proof}

\section{The problem with \BZ-unforgeability}\label{sec:BZisbroken}

Our search for a new definition of unforgeability for quantum-secure authentication is partly motivated by concerns about the \BZ security notion~\cite{BZ13a}. In this section, we make these concerns concrete by pointing out a significant security concern not addressed by this definition. Specifically, we demonstrate a MAC which is readily broken with an efficient attack, and yet is \BZ secure. The attack queries the MAC with a superposition over a particular subset $S$ of the message space, and then forges a valid tag for a message lying outside $S$.

One of the intuitive issues with \BZ is that it might rule out adversaries that have to measure, and thereby destroy, one or more post-query states to produce an interesting forgery. Constructing such an example seems not difficult at first. For instance, let us look at one-time \BZ, and construct a MAC from a \qPRF $f$ by sampling a key $k$ for $f$ and a superpolynomially-large prime $p$, and setting
\begin{equation}\label{eq:periodic-mac}
\Mac_{k, p}(m) = 
\begin{cases}
0^{n} &\text{ if } m = p,\\
(f_k(m \bmod p))  &\text{ otherwise.}
\end{cases}
\end{equation}
This MAC is forgeable: a quantum adversary can use a single query to perform period-finding on the MAC, and then forge at $0^n$. Intuitively, it seems plausible that the MAC is $1$-\BZ secure as period-finding uses a full measurement. This is incorrect for a somewhat subtle reason: identifying the hidden symmetry does not necessarily consume the post-query state completely, so an adversary can learn the period and a random input-output-pair of the MAC simultaneously. As shown in \expref{Lemma}{lem:double-spending} in \expref{Appendix}{sec:double-spending}, this is a special case of a fairly general situation, which makes establishing a proper \BZ ``counterexample'' difficult.

\subsection{A counterexample to \BZ}\label{sec:bz-kill}

Another intuitive problem with \BZ is that using the contents of a register can necessitate \emph{uncomputing} the contents of another one. We exploit this  insufficiency in the counterexample below.
Consider the following MAC construction.

\begin{construction}\label{con:real-BZ-killer}
	Given $k = (p, f, g,h)$ where $p\in\bits^{n}$ is a random period and $f,g,h: \bits^{n}\to\bits^{n}$ are random functions, define $M_k : \bits^{n+1} \to \bits^{2n}$ by
	\vspace{-.01cm}
	$$
	M_k(x) =
	\begin{cases}
	g(x' \bmod p)\|f(x')& x=1\|x'\,,\\
	0^n\|h(x')& x=0\|x', \ x'\neq p\,,\\
	0^{2n}& x=0\|p\,.
	\end{cases}
	$$
\end{construction}

Consider an adversary that queries as follows
\begin{equation}
\sum_{x, y} \ket{1, x}_X \ket{0^n}_{Y_1} \ket{y}_{Y_2}
\longmapsto 
\sum_{x, y} \ket{1, x}_X \ket{g_p(x)}_{Y_1} \ket{y \oplus f(x)}_{Y_2}\,,
\end{equation}
and then discards the first qubit and the $Y_2$ register; this yields $\sum_x \ket{x}\ket{g_p(x)}$. The adversary can extract $p$ via period-finding from polynomially-many such states, and then output $(0\|p,0^{2n})$. This attack only queries the MAC on messages starting with $1$ (e.g., ``\textsf{from Alice}''), and then forges at a message which starts with $0$ (e.g., ``\textsf{from Bob}.'') We emphasize that the forgery was never queried, not even with negligible amplitude. It is thus intuitively clear that this MAC does not provide secure authentication. And yet, despite this obvious and intuitive vulnerability, this MAC is in fact \BZ-secure.

\begin{theorem}\label{thm:real-BZ-killer}
	The MAC from \expref{Construction}{con:real-BZ-killer} is \BZ-secure.
\end{theorem}

We briefly summarize the proof idea before presenting the details.
The superposition oracle technique outlined in
\expref{Section}{sec:superoracles} achieves something that naively seems
impossible due to the quantum no-cloning theorem: it records on which
inputs the adversary has made non-trivial queries.\footnote{For the
  standard unitary oracle for a classical function, a query has no
  effect when the output register is initialized in the uniform
  superposition of all strings.} The information recorded in this way
cannot, in general be utilized in its entirety---after all, the
premise of the superposition oracle is that the measurement
$\mathcal M_F$ that samples the random function is delayed until after
the algorithm has finished, but it still has to be performed. Any
measurement $\mathcal M'$ that does not commute with $\mathcal M_F$
and is performed before $\mathcal M_F$, can disturb the outcome of
$\mathcal M_F$. If however, $\mathcal M'$ only has polynomially many
possible outcomes, that disturbance is at most inverse polynomial
according to \expref{Lemma}{lem:add-measure}.

Here, we sample the random function $f$ using a superposition oracle, and we chose to use a measurement $\mathcal M'$ to determine the \emph{number} of nontrivial queries that the adversary has made to $f$, which is polynomial by assumption. %
Random functions are \BZ-secure~\cite{BZ13a}, so the only way to break \BZ security is to output $(0\|p, 0^{2n})$ and $q$ other input-output-pairs. Querying messages that start with $0$ clearly only yields a negligible advantage in guessing $p$ by the Grover lower bound, so we consider an adversary querying only on strings starting with $1$. We distinguish two cases, either the adversary makes exactly $q$ non-trivial queries to $f$, or less than that. In the latter case, the success probability is negligible by the \BZ-security of $f$ and $h$. In the former case, we have to analyze the probability that the adversary guesses $p$ correctly. $f$ is not needed for that, so the superposition oracle register can be used to measure the set of $q$ queries that the adversary made. Using an inductive argument reminiscent of the hybrid method~\cite{BBBV} we show that this set is almost independent of $p$, and hence the period is equal to the difference of two of the queried inputs only with negligible probability. But if that is not the  case, the periodic version of $g$ is indistinguishable from a random function for that adversary which is independent of $p$.

\begin{proof}
  Let \adver be an adversary that makes $q$ quantum queries and
  outputs $q+1$ distinct candidate forgeries (where $q$ is selected by
  $\adver$ at runtime). We let this adversary interact with a mixed
  oracle, where $g, h$ and $p$ are treated as random variables, and
  $f$ is represented as a Fourier Oracle as in
  \expref{Section}{sec:superoracles}. We denote the relevant quantum
  registers as follows. First, the quantum oracle for $\Mac_k$ is a
  unitary operator on four registers: (i.) the $(n+1)$-qubit input
  register $X$, (ii.) the $n$-qubit output register $Y_1$ into which
  $g_p: x \mapsto g(x \mod p)$ is computed, and (iii.) the $n$-qubit
  output register $Y_2$ which, if the input $x$ starts with a $1$,
  interacts with the Fourier Oracle, which has (iv.) an $(n\cdot2^n)$-qubit register denoted by $F$, with the
  subregister corresponding to input $x\in\bits^{n}$ denoted by
  $F_x$. We set $Y = Y_1Y_2$. Finally, the workspace of $\adver$ is a $\poly$-qubit
  register denoted by $E$.

By the \BZ-unforgeability of random functions, any \BZ-adversary needs to output $(0\|p, 0^{2n})$ when successful, except with negligible probability. Indeed, suppose an adversary $\adver$ output $q+1$ input-output pairs of $M_k$, none of which is equal $(0\|p, 0^{2n})$ with noticeable probability. Then  we can use that adversary to construct a \BZ-adversary against $\tilde M_k$ defined as
$$
\tilde M_k(x) =
\begin{cases}
f(x')& x=1\|x'\,,\\
0^n\|h(x')& x=0\|x'
\end{cases}
$$
by simulating an  $M_k$-oracle for $\adver$ using the oracle for $\tilde M_k$. To learn $p$, an adversary that makes a polynomial number of queries needs to use messages starting with 1, as the lower bound for unstructured search \cite{BBBV} implies that querying messages starting with $0$ only provides negligible advantage for learning $p$. We will thus prove in the following that an adversary whose queries are entirely supported on the space of messages starting with $1$ cannot succeed. The proof for a general adversary is similar, if more laborious. We thus omit $h$ from the description in the following, focusing on the task of outputting $q$ input-output pairs of $f$ and the period $p$.

	Let $\ket\psi_{XYEF}$ denote the final state of $\adver$ and the Fourier Oracle, after the $q+1$ candidate forgeries have been measured, but prior to any other measurements. Recall that each ``number projector'' $P_l$ from \expref{Section}{sec:superoracles} projects $F$ to the subspace spanned by basis states with exactly $l$ non-zero entries. We apply to $\ket{\psi}$ the two-outcome measurement defined by $P_{<q}=\sum_{l=0}^{q-1}P_l$ and its complementary projector $P_{\ge q}=\one-P_{<q}$, effectively measuring whether $F$ contains fewer than $q$ non-zero entries (i.e., registers $F_x$ containing a state other than $0^n$); note that it cannot contain more than $q$ by \expref{Lemma}{lem:Zhancrement}. By \expref{Lemma}{lem:add-measure}, applying this measurement decreases the success probability of $\algo A$ at any particular task by a factor $1/2$. We handle the two possible outcomes ($<q$ and $q$) separately.
	
	\vspace{4pt}\noindent\textbf{Case $<q$:}~Let $\ket{\psi^{<q}}_{XYEF}:=P_{<q}\ket\psi_{XYEF}$ be the post-measurement state. Note that $P_l \ket{\psi^{<q}} = 0$ for all $l \geq q$, i.e., each basis component of $\ket{\psi^{<q}}$ has fewer than $q$ non-zero entries in $F$. On the other hand, the output of $\adver$ contains at least $q$ candidate input-output pairs $(x_i, y_i)$ of $f$ (since $(0\|p, 0^{2n})$ is the only input-output pair of $\Mac_k$ that does not also contain an input-output pair of a random function). We apply the $q$-outcome measurement to $F$ which asks: ``among the registers $\{F_{x_i}\}_{i=1}^q$, which is the first one to contain $0^n$?'' This measurement is defined by projectors
	\[
	\Pi_j := \bigotimes_{i=1}^j \left(\one - \proj{0^n}\right)_{F_{x_i}} \otimes \proj{0^n}_{F_{x_j}}\,.
	\]
	Adding this measurement to $\adver$ ensures that $F_{x_j}$ is in the state $0^n$ for some $j$, at the cost of multiplying $\adver$'s success probability by $1/q$ (by \expref{Lemma}{lem:add-measure}). Recalling that, in the Fourier Oracle picture, $f(x_j)$ is the result of QFT-ing and then fully measuring $F_{x_j}$, we see that $f(x_j)$ is now uniformly random and independent of $y_j$. The original $\adver$ (i.e., without the measurement $\{\Pi_j\}_j$) thus succeeded with probability at most $q \cdot 2^{-n}$. \footnote{This argument amounts to an alternative proof that random functions are \BZ-secure.}
	
	\vspace{4pt}\noindent\textbf{Case $q$:}~We will denote the post-measurement state in this case by $\ket{\psi^q_{g_p}} := P_q \ket{\psi}$, emphasizing that the state was produced by interacting with the oracle $g_p$. By the \BZ-security of $f$ (\expref{Theorem}{thm:Rand-BZ}) it suffices to show that the correct period $p$ is output by \adver (by measuring, say, some designated subregister of $E$ of the state $\ket{\psi^q_{g_p}}$) with at most negligible probability. %
	Since testing success at outputting $p$ does not involve the register $F$, we are free to apply any quantum channel to the $F$ register of $\ket{\psi^q_{g_p}}$. We choose to measure which $q$ subregisters of $F$ are in a non-zero state. This projective measurement is defined by projectors
	\begin{equation}
	P_K= \bigotimes_{x \in K} \left(\one - \proj{0^n}\right)_{F_x} \otimes \bigotimes_{x \notin K} \proj {0^n}_{F_x}
	\qquad \text{and} \qquad
	P_{\text{rest}}=\one-\sum_K P_K\,,
	\end{equation}
	where $K\subset\bits^n$ with $|K|=q$. Note that $P_\text{rest} = \one - P_q$, so the outcome ``rest'' never occurs for $\ket{\psi^q_{g_p}}$.  In the following we denote by $\mathbf K$ the random variable obtained from this measurement. We also set some other random variables in boldface to better distinguish them from particular values they can take.
	
	Now consider the preparation of the state $\ket{\psi^q}$ (by $\adver$ and the Fourier Oracle) with an arbitrary choice of oracle function $h: \bits^n \to \bits^n$ in place of $g_p$.  We will denote this state by $\ket{\psi_h^q}$. We now show that, conditioned on a particular measurement outcome $K$, we can arbitrarily relabel the values of $h$ outside $K$, without affecting the output state of the algorithm.

	\begin{lemma}\label{lem:g-projected}
		Let $K \subset \bits^n$ with $|K| = q$ and $h, h': \bits^n \to \bits^n$ a pair of functions satisfying $h(x) = h'(x)$ for all $x \in K$. Then $P_K\ket{\psi_{h}^q} = P_K \ket{\psi_{h'}^q}$.
	\end{lemma}
	
	\begin{proof}
		Let $W^{(j)}_{XYEF} := V^{(j)}_{XYE}U^{(h)}_{XY_1}U^{\mathrm{FO}}_{XY_2F}$, where $V^{(j)}$ is \adver's $j$-th internal unitary, $U^{(h)}$ is the standard oracle unitary for $h$, and $U^{\mathrm{FO}}$ is the Fourier Oracle unitary as described in \expref{Section}{sec:superoracles}. The intermediate states are
		\begin{equation}
		\ket{\varphi_{h,k}}_{XYEF} := W^{(k)} \cdots W^{(1)} V^{(0)} \ket{0}_{XYEF}\,,
		\end{equation}
		and the final state is $\ket{\psi_h} := \ket{\varphi_{h,q}}$. By \expref{Lemma}{lem:Zhancrement}, $P_l\ket{\varphi_{k,h}}=0$ for all $l>k$, so 
		$$
		\ket{\psi^q_h} = P_q \ket{\psi_h} = P_q W^{(q)} \cdots W^{(k+1)} \ket{\varphi_{k,h}}
		= \sum_{l=0}^k P_q W^{(q)} \cdots W^{(k+1)} P_l \ket{\varphi_{k,h}}\,.
		$$
		For the $l$ term in the sum above, the unitary applies $q-k$ queries to $P_l\ket{\varphi_{k,h}}$; by \expref{Lemma}{lem:Zhancrement} this term is thus zero unless $l = k$. We can therefore insert a $P_k$ after the $k$-th query for free when projecting with $P_q$ in the end. Explicitly,
		\begin{equation}
		\ket{\psi^q_h} = P_q W^{(q)} P_{q-1} W^{(q-1)} P_{q-2} \cdots P_1 W^{(1)} V^{(0}) \ket{0}_{XYEF}\,.
		\end{equation}
		
		We first show that we can apply
		$$
		\tilde P_K : = \bigotimes_{x \in K} \one_{F_x} \otimes \bigotimes_{x \in K^c} \proj {0^n}_{F_x}
		$$
		after every query of \adver. 
		
		We are interested in the state $P_K\ket{\psi}_{XYEF}=P_KP_q\ket{\psi}_{XYEF}$. We can make a similar argument as above to show that we can project with $\tilde P_K$ after every query as well. As the FO-unitary is the only one that acts on $F$, and because $\tilde P_K\ket 0^{\otimes n2^n}=\ket 0^{\otimes n2^n}$, we can even apply the projector $\tilde P_K$ before and after each query. We write $N=N_K+N_{K^c}$, where
		\begin{equation}
		N_K=\sum_{x\in K}(\one-\proj 0)_{F_x}\otimes \one^{\otimes (2^n-1)},
		\end{equation}
		i.e., $N_K$ and $N_{K^c}$ measure the number of non-zero entries inside and outside $K$, respectively. \expref{Lemma}{lem:Zhancrement} applies to $N_K$ and $N_{K^c}$ separately, and $ P_KN_K\ket{\psi}_{XYEF}=N_KP_K\ket{\psi}_{XYEF}=qP_K\ket{\psi}_{XYEF}$. Therefore we have, defining
		\begin{equation}
		U_{>k}=V^{(q)}_{XYE}U^{(h)}_{XY_1}U^{\mathrm{FO}}_{XY_2F}V^{(q-1)}_{XYE}U^{(h)}_{XY_1}U^{\mathrm{FO}}_{XY_2F}...V^{(k+1)}_{XYE}U^{(h)}_{XY_1}U^{\mathrm{FO}}_{XY_2F}V^{(k)}_{XYE}
		\end{equation}
		and using the same argument as above, that 
		\begin{equation}
		P_KU_{>k}N\ket{\psi^k}=P_KU_{>k}N_K\ket{\psi^k}=kP_KU_{>k}\ket{\psi^k},
		\end{equation}
		and hence
		\begin{equation}
		P_KU_{>k}N_{K^c}\ket{\psi^k}=P_KU_{>k}N\ket{\psi^k}-P_KU_{>k}N_{K}\ket{\psi^k}=0,
		\end{equation}
		implying $N_{K^c}\ket{\psi^k}=0$. But the projector onto the zero-eigenspace of $N_{K^c}$ is $\tilde P_K$, so $\tilde P_K\ket{\psi^k}=\ket{\psi^k}$.
		
		With an even simpler argument we can insert a projector $P^{\neq 0}_{Y_2}=\left(\one-\proj 0\right)_{Y_2}$ before every query. This is because $U^{\mathrm{FO}}\ket 0_{Y_2}\ket{\gamma}_{XF}=\ket 0_{Y_2}\ket{\gamma}_{XF}$, and therefore the number operator eignenvalue does not increase.
		
		To show that $U^{(h)}\tilde P_KU^{\mathrm{FO}}\left(P^{\neq 0}_{Y_2}\otimes (\tilde P_K)_F\right)$ is independent of the values outside $K$, we observe that for all $x\notin K$, $y\in\bits^n\setminus\{ 0^n\}$ and for all states $\ket\gamma_{Y_1EF}$, we have 
		\begin{align}
		&U^{(g,p)}\tilde P_KU^{\mathrm{FO}}\left(P^{\neq 0}_{Y_2}\otimes (\tilde P_K)_F\right)\ket x_X\otimes \ket{\phi_y}_{Y_2}\otimes \ket\gamma_{Y_1EF}\nonumber\\
		&=U^{(g,p)}\ket x_X\otimes\left(\tilde P_K\left(H^{\otimes n}\right)_{Y_2}\CNOT_{Y_2:F_x} \ket y_{Y_2}\otimes \tilde P_K\ket\gamma_{Y_1EF}\right)\nonumber\\
		&=U^{(g,p)}\ket x_X\otimes\left(\tilde P_K\proj 0_{F_x}\left(H^{\otimes n}\right)_{Y_2}\CNOT_{Y_2:F_x}\proj 0_{F_x} \ket y_{Y_2}\otimes \tilde P_K\ket\gamma_{Y_1EF}\right)\nonumber\\
		&=U^{(g,p)}\ket x_X\otimes\left(\tilde P_K\proj 0_{F_x}\ketbra{y}{0}_{F_x}\otimes \ket{\phi_y}_{Y_2}\otimes \tilde P_K\ket\gamma_{Y_1EF}\right)\nonumber\\
		&=0,\label{eq:basic-calculation}
		\end{align}
		where we have used that for all $x\notin K$ it holds that $\proj 0_{F_x}\tilde P_K=\tilde P_K$. This implies that our artificial oracle $U^{(g,p)}\tilde P_KU^{\mathrm{FO}}\left(P^{\neq 0}_{Y_2}\otimes (\tilde P_K)_F\right)$ (together with a renormalization) only gives \adver access to $g(x\mod p)$ for inputs $x\in K$.%
		
		This concludes the proof of \expref{Lemma}{lem:g-projected}.
	\end{proof}
	
	We now continue with the ``\textbf{case $q$}'' proof of the theorem. We bound \adver's success probability separately for each outcome $K$. Indeed, it suffices to show that for all $K\subset\bits^n$, $|K|=q$ the probability that the output contains a pair $(0\|p,0^{2n})$ is negligible if \adver continues with 
	\begin{equation}
	\ket{\psi^{q, K}}:=\frac{P_K\ket{\psi^{q}}}{\left\|P_K\ket{\psi^{q}}\right\|_2}
	\end{equation}
	in place of $\ket\psi$.

	We show that the periodic oracle can be replaced by a non-periodic one, except with negligible probability. More precisely, if $p'$ is \adver's output, there exists an event $E$ such that $\Pr[E]=1-\negl$ and $\Pr[p'=p_0|E, p=p_0]=\Pr[p'=p_0|E, p=0]$ for all $p_0\in \bits^n$. 
	In the following, let us denote the oracle for the MAC of \expref{Construction}{con:real-BZ-killer-summary} with functions $f$ and $g$ and period $p$ by $\mathcal O_{f,g_p}$. We define
	\begin{equation}
	\mathcal P^{\text{bad}}_K=\left\{p\in\bits^n\Big|\exists x,x'\in K: \,p|x-x'\right\}.
	\end{equation}
	For $K\subset \bits^n$ and $p\in\bits^n$, if $p\notin\mathcal P^{\text{bad}}_K$, let $T_{K,p}\subset\bits^n$ be a transversal for $p$ (i.e., a maximal set such that for $x,y\in T_{K,p}$ it holds that $x\neq y\mod p$) such that $T_{K,p}\cap K=K$. Using this transversal, we can define for each $K$ a random periodic function $g^{(K)}_p$ that is identically distributed with $g_p$, as follows. 
	\begin{itemize}
		\item If $p \in\mathcal P^{\text{bad}}_K$, we set $g_p^{(K)}(x) = g(x \mod p)$. 
		\item If $p\notin\mathcal P^{\text{bad}}_K$, we set $g^{(K)}_{p}(x) = g(y)$ for $y \in T_{K,p}$ such that $x=y\mod p$. 
	\end{itemize}
	
	For a unitary algorithm $\tilde \adver$ that makes $\ell$ queries to an oracle $\mathcal O_{f, g_p}$, we define  the following procedures:
	\smallskip
	
	\noindent
	\textbf{Procedure 0}
	\begin{enumerate}
		\item Sample $f$, $g$ and $p$.
		\item Run $\tilde \adver$ with oracle $\mathcal O_{f,g_p}$ resulting in a final adversary-oracle state $\ket{\hat\psi}$. Apply the measurement $\{P_{\ge \ell}, P_{<\ell}\}$ to $F$. If outcome is $<\ell$, output ``fail.''
		\item Measure $K$. If $p\in\mathcal P^{\text{bad}}_K$, output ``bad.'' Otherwise, let $\ket \psi$ be the post-measurement state of adversary and oracle, i.e., $\ket \psi=P_KP_{\ge \ell}\ket{\hat\psi}=P_K\ket{\hat\psi}$.
		\item Output $(K,p,\ket \psi)$.
	\end{enumerate}
	\textbf{Procedure $\mathbf{0}_{\mathbf{K}}$}\vspace{7pt}
	
	Same as Procedure 0, except with oracle $\mathcal O_{f, g^{(K)}_p}$ instead of $\mathcal O_{f, g_p}$.\\
	\textbf{Procedure 1}
	\begin{enumerate}
		\item Sample $f$ and $g$.
		\item Run \adver with an oracle $\mathcal O_{f,g_0}$ resulting in a final adversary-oracle state $\ket{\hat\psi}$. Apply the measurement $\{P_{\ge \ell}, P_{<\ell}\}$ to $F$. If outcome is $<\ell$, output ``fail.''
		\item Measure $K$ and sample $p$. If $p\in\mathcal P^{\text{bad}}_K$, output ``bad.'' Otherwise, let $\ket \psi$ be the post-measurement state of adversary and oracle, i.e., $\ket \psi=P_KP_{\ge \ell}\ket{\hat\psi}=P_K\ket{\hat\psi}$.
		\item Output $(K,p,\ket \psi)$.
	\end{enumerate}
	We first observe that for all $K$, the outputs of procedures $0$ and $0_K$ are identically distributed because $g_p$ and $g_{p,K}$ are. Note that for any fixed $K$, $P_K P_q = P_K$; this, together with \expref{Lemma}{lem:g-projected}, implies that
	\begin{equation}
	\Pr\left[(K,p,\ket \psi)\leftarrow\mathrm{Procedure}\ 0_K\right]=		\Pr\left[(K,p,\ket \psi)\leftarrow\mathrm{Procedure}\ 1\right].
	\end{equation}
	It follows that, still for a fixed $K$,
	\begin{equation}\label{eq:procedures-equal}
	\Pr\left[(K,p,\ket \psi)\leftarrow\mathrm{Procedure}\ 0\right]=		\Pr\left[(K,p,\ket \psi)\leftarrow\mathrm{Procedure}\ 1\right].
	\end{equation}
	This implies also that in any of the three procedures, conditioned on the event that the output is neither ``fail'' nor ``bad'' and on a fixed first output $K$, $p$ is uniformly distributed on $\bits^n\setminus\mathcal P_{\mathrm{bad}}$. In other words,
	\begin{equation}\label{eq:stilluniform}
	\Pr\bigl[\mathbf p=p\,|\,\mathbf K=K\wedge \mathbf p\notin \mathcal P^{\text{bad}}_K\bigr]=\begin{cases}
	\left(2^n-|\mathcal P^{\text{bad}}_K|\right)^{-1}& p\notin \mathcal P^{\text{bad}}_K,\\
	0&\text{else}.
	\end{cases}
	\end{equation}
	Let us denote the event that a procedure outputs a triple $(K,p,\ket \psi)$ by ``good.'' 
	
	In what follows, we fix a particular period $p$, an outcome of the period-sampling step (step 1 in Procedures $0$ and $0_K$ and step 3 in Procedure $1$). Given a number $\ell$ of queries we identify three subspaces of $\Hi_F$ corresponding to the three outcomes ``good,'' ``bad'' and ``fail'' of the procedures above:
	\begin{align}
	S^{\ell}_{\mathrm{fail}}&=\mathrm{range}(P_{<\ell})\,,\\
	S^{\ell}_{\mathrm{bad}}&=\mathrm{span}\left\{\mathrm{range}\left(P_K\right)\Big|K\subset \bits^n,\ |K|=\ell,\ \exists x,y\in K: p|x-y\right\}, \text{ and}\\
	S^{\ell}_{\mathrm{good}}&=\left(S^{\ell}_{\mathrm{fail}}\right)^\perp\cap\left(S^{\ell}_{\mathrm{bad}}\right)^\perp.
	\end{align}
	We emphasize that the decomposition defined by these subspaces depends on the aforementioned period $p$. We let $P_i^\ell$ for $i\in \{\mathrm{good}, \mathrm{bad}, \mathrm{fail}\}$ denote the projectors onto these subsets.

	By the above reasoning we know that for any algorithm that makes $\ell$ queries to an oracle $\mathcal O$ and has final state $\ket{\psi_{\mathcal O}^\ell}_{AF}$, it holds that $P^\ell_{\mathrm{good}}\ket{\psi_{\mathcal O_{f,g_p}}^\ell}_{AF}=P^\ell_{\mathrm{good}}\ket{\psi_{\mathcal O_{f,g_0}}^\ell}_{AF}$. It is easy to see that when another query is made, i.e., the $\ell+1$st query of some algorithm, some transitions from $S^{\ell}_i$ to $S^{\ell+1,p}_j$ are impossible. We only need one impossibility, namely that according to \expref{Lemma}{lem:Zhancrement}, $P_i^{\ell+1} U^{\mathrm{FO}}P_{\mathrm{fail}}^\ell=0$ for all $i\neq\mathrm{fail}$. In words, once an adversary has fallen behind his $q$-query plan of making one non-trivial query to $f$ in every query, he can never catch up. Also note that for $\ell=0$, $S^{\ell}_{\mathrm{fail}}=S^{\ell}_{\mathrm{bad}}=0$. It is now easy to show by induction that for a $q$-query adversary \adver with final adversary-oracle state $\ket\phi$ it holds that
	\begin{equation}\label{eq:bad-decomposition}
	\left\|P^q_{\mathrm{bad}}\ket{\phi}\right\|_2\le\sum_{\ell=1}^q\left\|P^\ell_{\mathrm{bad}}U^{\mathrm FO}P^{\ell-1}_{\mathrm{good}}\ket{\phi_\ell}\right\|_2,
	\end{equation}
	where $\ket{\phi_\ell}$ is the adversary oracle state before the $\ell$th query. The induction step is proven as follows. Assume the above formula is true for $q$. Then we have for a $(q+1)$-query adversary \adver with final adversary-oracle state $\ket{\phi}$
	\begin{align}
	\left\|P^{q+1}_{\mathrm{bad}}\ket\phi\right\|_2&=\left\|P^{q+1}_{\mathrm{bad}}\ket{\psi_{q+1}}\right\|_2\\
	&\le \left\|P^{q+1}_{\mathrm{bad}}U^{\mathrm FO}P^{q}_{\mathrm{good}}\ket{\phi_{q+1}}\right\|_2+\left\|P^{q+1}_{\mathrm{bad}}U^{\mathrm FO}P^{q}_{\mathrm{bad}}\ket{\phi_{q+1}}\right\|_2\\
	&\quad+\left\|P^{q+1}_{\mathrm{bad}}U^{\mathrm FO}P^{q}_{\mathrm{fail}}\ket{\phi_{q+1}}\right\|_2\\
	&=\left\|P^{q+1}_{\mathrm{bad}}U^{\mathrm FO}P^{q}_{\mathrm{good}}\ket{\phi_{q+1}}\right\|_2+\left\|P^{q+1}_{\mathrm{bad}}U^{\mathrm FO}P^{q}_{\mathrm{bad}}\ket{\phi_{q+1}}\right\|_2\\
	&\le\left\|P^{q+1}_{\mathrm{bad}}U^{\mathrm FO}P^{q}_{\mathrm{good}}\ket{\phi_{q+1}}\right\|_2+\left\|U^{\mathrm FO}P^{q}_{\mathrm{bad}}\ket{\phi_{q+1}}\right\|_2\\
	&=\left\|P^{q+1}_{\mathrm{bad}}U^{\mathrm FO}P^{q}_{\mathrm{good}}\ket{\phi_{q+1}}\right\|_2+\left\|P^{q}_{\mathrm{bad}}\ket{\psi_{q}}\right\|_2.
	\end{align}
	Here we have used the unitary invariance of the Euclidean together with the observation that the state $\ket{\phi}$ is obtained from the state $\ket{\psi_{q+1}}$ right after the $(q+1)$-st query of \adver by a unitary acting on the adversary's space only and which therefore commutes with $P^{q}_{\mathrm{bad}}$ in the first, the triangle inequality in the second line, the observation that $P_i^{\ell+1} U^{\mathrm{FO}}P_{\mathrm{fail}}^\ell=0$ in the third line, and the fact that $\|P\|_\infty\le 1$ for any projector $P$ in the fourth line. In the fifth line we use the same argument as in the first line, just for $\ket{\phi_{q+1}}$ and $\ket{\psi_{q}}$. This proves Equation \eqref{eq:bad-decomposition}.
	
	It remains to bound
        \[
          \left\|P^\ell_{\mathrm{bad}}U^{\mathrm FO}P^{\ell-1}_{\mathrm{good}}\ket{\phi_\ell}\right\|_2\,.
        \]
        To this end, suppose that we measure the $X$-register of
        $\ket{\phi_\ell}$ in the computational basis with outcome
        $ \mathbf X_\ell$, as well as $\mathbf K^{(\ell-1)}$ the set
        of nonzero registers in $F$. According to Equations
        \eqref{eq:procedures-equal} and \eqref{eq:stilluniform}, we
        have that $\mathbf X_\ell$ and $\mathbf p$ are independent and
        $\mathbf p$ is uniformly distributed on
        $\bits^n\setminus \mathcal P^{\mathrm{bad}}_{K}$ conditioned
        on $\mathbf p\notin\mathcal P^{\mathrm{bad}}_{K}$ and
        $\mathbf K=K$ for a fixed $(\ell-1)$-element set $K$. It
        follows that
	\begin{align}
	&\Pr\left[\mathbf p\in\mathcal P^{\mathrm{bad}}_{\mathbf K\cup \{\mathbf X_\ell\}} \Big|\mathbf K=K\wedge \mathbf p\notin P^{\mathrm{bad}}_{K}\right]\\
	&=\Pr\left[\exists y\in K: \mathbf p|(\mathbf X_\ell-y)\Big|\mathbf K=K\wedge\mathbf p\notin P^{\mathrm{bad}}_{K}\right]\\
	&\le \frac{(\ell-1)2^{\frac{c' n}{\log n}}}{2^n-\frac{(\ell-1)(\ell-2)}{2}2^{\frac{c' n}{\log n}}}\le (\ell-1)2^{-n\left(1-\frac{c }{\log n}\right)}\,.\label{eq:bound-given-stuff}
	\end{align}
	Here the last inequality holds for some $0<c<c'$ and large enough $n$, and we have used in the third line that there exists a constant $c'>0$ such that the number of divisors of an integer $M$ is bounded by $2^{c\frac{\log M}{\log\log M}}$ which also implies
	\begin{equation}
	\left|\mathcal{P}^{\text{bad}}_K\right|\le \frac{(\ell-1)(\ell-2)}{2}2^{c\frac{n}{\log n}}
	\end{equation}
	for all $K\subset\bits^n$, $|K|=\ell$.
	We would now like to relate the above probability to
	\[
	\mathbb E\left[\left\|P^\ell_{\mathrm{bad}}U^{\mathrm FO}P^{\ell-1}_{\mathrm{good}}\ket{\phi_\ell}\right\|_2^2\right]\,.
	\]
	To this end we analyze how the operator $P^\ell_{\mathrm{bad}}U^{\mathrm FO}P^{\ell-1}_{\mathrm{good}}$ behaves on states of the form $\ket x_X\otimes\ket{\phi_y}\otimes  \ket \zeta_{EF}$ such that $\left(P_{K}\right)_F\ket \zeta_{EF}=\ket \zeta_{EF}$ for some fixed $K\not \ni x$ and $p\in\bits^n$ such that $p\not\in\mathcal{P}^{\text{bad}}_{K}$ . We calculate
	\begin{align}
	&U^{\mathrm FO}P^{\ell-1}_{\mathrm{good}}\ket x_X\otimes \ket{\phi_y}\otimes  \ket \zeta_{EF}\\
	&=U^{\mathrm FO}\ket x_X\otimes \ket{\phi_y}\otimes  \ket \zeta_{EF_{K}}\otimes \ket{0^{n(2^n-\ell+1)}}_{F_{K^c}}\\
	&=\left(H^{\otimes n}\right)_Y\CNOT_{Y:F_x}\ket x_X\otimes \ket{y}\otimes  \ket \zeta_{EF_{K}}\otimes \ket{0^{n(2^n-\ell+1)}}_{F_{K^c}}\\
	&=\left(H^{\otimes n}\right)_Y\ket x_X\otimes \ket{y}\otimes  \ket \zeta_{EF_{K}}\otimes\ket{y}_{F_x} \otimes \ket{0^{n(2^n-\ell)}}_{F_{\left({K}\cup \{x\}\right)^c}}\\
	&=\ket x_X\otimes \ket{\phi_y}\otimes  \ket \zeta_{EF_{K}}\otimes\ket{y}_{F_x} \otimes \ket{0^{n(2^n-\ell)}}_{F_{\left({K}\cup \{x\}\right)^c}}.
	\end{align}
	In the first equation we have use the assumptions that $\left(P_{K}\right)_F\ket \zeta_{EF}=\ket \zeta_{EF}$ and $p\not\in\mathcal{P}^{\text{bad}}_{K}$; the rest of the calculation is analogous to Equation \eqref{eq:basic-calculation}. This implies that
	\begin{align}
	P_{K\cup\{x\}}U^{\mathrm FO}P^{\ell-1}_{\mathrm{good}}\ket x_X\otimes \ket{\phi_y}\otimes  \ket \zeta_{EF}=U^{\mathrm FO}P^{\ell-1}_{\mathrm{good}}\ket x_X\otimes \ket{\phi_y}\otimes  \ket \zeta_{EF}
	\end{align}
	and therefore 
	\begin{align}
	&P^\ell_{\mathrm{bad}}U^{\mathrm FO}P^{\ell-1}_{\mathrm{good}}\ket x_X\otimes \ket{\phi_y}\otimes  \ket \zeta_{EF}\\
	&=\begin{cases}
	U^{\mathrm FO}P^{\ell-1}_{\mathrm{good}}\ket x_X\otimes \ket{\phi_y}\otimes  \ket \zeta_{EF}& \text{ if }\exists x'\in K: p|(x-x')\,,\\
	0&\text{ otherwise.}
	\end{cases}\label{eq:Pbad-simplify}
	\end{align}
	We therefore calculate for a fixed $p$,
	\begin{align*}
	\left\|P^\ell_{\mathrm{bad}}U^{\mathrm FO}P^{\ell-1}_{\mathrm{good}}\ket{\phi_\ell}\right\|_2^2
	&=\Bigg\|\sum_{\substack{K\subset\bits^n\\ |K|=\ell-1\\p\not\in\mathcal{P}^{\text{bad}}_{K}}}\sum_{x\in\bits^n}P^\ell_{\mathrm{bad}}U^{\mathrm FO}\left(\proj x_X\otimes P_{K}\right)\ket{\phi_\ell}\Bigg\|_2^2\\
	&=\Bigg\|U^{\mathrm FO}\sum_{\substack{K\subset\bits^n\\ |K|=\ell-1\\p\not\in\mathcal{P}^{\text{bad}}_{K}}}\sum_{\substack{x\in\bits^n\setminus K\\ \exists x'\in K: p|(x-x')}}\left(\proj x_X\otimes P_{K}\right)\ket{\phi_\ell}\Bigg\|_2^2\\
	&=\sum_{\substack{K\subset\bits^n\\ |K|=\ell-1\\p\not\in\mathcal{P}^{\text{bad}}_{K}}}\sum_{\substack{x\in\bits^n\setminus K\\ \exists x'\in K: p|(x-x')}}\left\|\left(\proj x_X\otimes P_{K}\right)\ket{\phi_\ell}\right\|_2^2\\
	&=\Pr\left[p\notin\mathcal{P}^{\text{bad}}_{\mathbf K}\wedge p\in \mathcal{P}^{\text{bad}}_{\mathbf K\cup \{\mathbf X_\ell\}}\Big|\mathbf p=p\right]\,.
	\end{align*}
	Using Equation \eqref{eq:bound-given-stuff} we can bound
	\begin{align*}
	&\mathbb E_{p\leftarrow \bits^n}\left[\Pr\left[p\notin\mathcal{P}^{\text{bad}}_{\mathbf K}\wedge p\in \mathcal{P}^{\text{bad}}_{\mathbf K\cup \{\mathbf X_\ell\}}\right]\right]\\
	&=\Pr\left[\mathbf p\notin\mathcal{P}^{\text{bad}}_{\mathbf K}\wedge \mathbf p\in \mathcal{P}^{\text{bad}}_{\mathbf K\cup \{\mathbf X_\ell\}}\right]\\
	&=\sum_{K\subset \bits^n}\Pr\left[\mathbf p\notin\mathcal{P}^{\text{bad}}_{K}\wedge \mathbf p\in \mathcal{P}^{\text{bad}}_{K\cup \{\mathbf X_\ell\}}\Big|\mathbf K=K_0\right]\Pr[\mathbf K=K_0]\\
	&=\sum_{K\subset \bits^n}\Pr\left[\mathbf p\in \mathcal{P}^{\text{bad}}_{K\cup \{\mathbf X_\ell\}}\Big|\mathbf K=K_0\wedge \mathbf p\notin\mathcal{P}^{\text{bad}}_{\mathbf K}\right]\Pr\left[\mathbf p\notin \mathcal{P}^{\text{bad}}_{K}\wedge \mathbf K=K\right]\\
	&\le \Pr\left[\mathbf p\notin \mathcal{P}^{\text{bad}}_{\mathbf K}\right](\ell-1)2^{-n\left(1-\frac{c }{\log n}\right)}\\
	&\le (\ell-1)2^{-n\left(1-\frac{c }{\log n}\right)}\,.
	\end{align*}
	Here we have used Equation \eqref{eq:bound-given-stuff} in the first inequality. The probability in the first line is taken over a run of the adversary with a fixed period and random $g$ and $f$, and in the other lines the period is picked uniformly at random from $\bits^n$ as for a properly generated key in \expref{Construction}{con:real-BZ-killer-summary}. The last two equations together imply
	\begin{align}
	\mathbb E\left[\left\|P^\ell_{\mathrm{bad}}U^{\mathrm FO}P^{\ell-1}_{\mathrm{good}}\ket{\phi_\ell}\right\|_2^2\right]&\le (\ell-1)2^{-n\left(1-\frac{c }{\log n}\right)}.
	\end{align}
	Plugging this into Equation \eqref{eq:bad-decomposition} yields
	\begin{align}
	\Pr\left[\mathbf p\in \mathcal P^{\text{bad}}_{\mathbf K}\right]&=\mathbb E\left[\left\|P^q_{\mathrm{bad}}\ket{\phi}\right\|_2^2\right]\nonumber\\
	&\le  \mathbb E\left[\left(\sum_{i=1}^q\left\|P^\ell_{\mathrm{bad}}U^{\mathrm FO}P^{\ell-1}_{\mathrm{good}}\ket{\phi_\ell}\right\|_2\right)^2\right]\nonumber\\
	&\le q \sum_{i=1}^q \mathbb E\left[\left\|P^\ell_{\mathrm{bad}}U^{\mathrm FO}P^{\ell-1}_{\mathrm{good}}\ket{\phi_\ell}\right\|_2^2\right]\nonumber\\
	&\le\left(\sum_{\ell=1}^q\sqrt{(\ell-1)2^{-n\left(1-\frac{c }{\log n}\right)}}\right)^2\nonumber\\
	&\le \frac{q^2(q-1)}{2}2^{-n\left(1-\frac{c }{\log n}\right)}\label{eq:badbound}
	\end{align}
	using the Cauchy-Schwartz inequality in the second line.
	This finally implies that the adversary's guess $p'$ is equal to $p$ and the measurement $<q$ vs. $\ge q$ returns $\ge q$ with probability at most
	\begin{align}
	&	\Pr[\mathbf p=\mathbf p'\wedge \text{``$\ge q$''}]\\
	\le& \Pr\left[\mathbf p\in\mathcal{P}^{\text{bad}}_{\mathbf K}\wedge \text{``$\ge q$''} \right]+ \Pr\left[\mathbf p\notin\mathcal{P}^{\text{bad}}_{\mathbf K} \wedge \mathbf p=\mathbf p'\wedge \text{``$\ge q$''}\right]\\
	\le&\Pr\left[\mathbf p\in\mathcal{P}^{\text{bad}}_{\mathbf K}\wedge\text{``$\ge q$''}\right]+ \Pr\left[\mathbf  p=\mathbf p'\big|\mathbf p\notin\mathcal{P}^{\text{bad}}_{\mathbf K}\wedge\text{``$\ge q$''}\right]\\
	\le&\frac{q^2(q-1)}{2}2^{-n\left(1-\frac{c }{\log n}\right)}+\left(2^n-\frac{(\ell-1)(\ell-2)}{2}2^{\frac{c' n}{\log n}}\right)^{-1}\\
	\le& \negl.
	\end{align}
	Here we have used Equation \eqref{eq:badbound} and the uniformity of $\mathbf p $ conditioned on $\mathbf p\notin\mathcal{P}^{\text{bad}}_{\mathbf K}$ and $\mathbf K=K$ in the  last line.
\end{proof}

\paragraph{Remark.}  
It's not hard to see that the MAC from \expref{Construction}{con:real-BZ-killer} is not \GYZ-secure. Indeed, observe that the forging adversary described above queries on messages starting with $0$ only, and then forges successfully on a message starting with $1$. If the scheme was \GYZ secure, then in the accepting case, the portion of this adversary between the query and the final output would have a simulator which leaves the computational basis invariant. Such a simulator cannot change the first bit of the message from $0$ to $1$, a contradiction.

By  \expref{Theorem}{thm:buinsecure}, this \BZ-secure MAC is also not \BU-secure.

\begin{corollary}\label{cor:real-BZ-killer-BU}
	The MAC from \expref{Construction}{con:real-BZ-killer} is \BU-insecure.
\end{corollary}

\appendix
\section{Technical proofs}

\subsection{The Fourier Oracle number operator}\label{sec:proof-zhancrement}

We now restate and prove \expref{Lemma}{lem:Zhancrement}.
\numberoplem*

\begin{proof}
	Let $\ket\psi_{XYEF}$ be an arbitrary query state, where $X$ and $Y$ are the query input and output registers, $E$ is the algorithm's internal register and $F$ is the FO register. We expand the state in the computational basis of $X$,
	\begin{equation}
		\ket\psi_{XYEF}=\sum_{x \in \bits^n}p(x)\ket x_X\ket{\psi_x}_{YEF}.
	\end{equation}
	Now observe that
\[
          U^{\mathrm{FO}}_{XYF}\ket x_X\ket{\psi_x}_{YEF}=\ket x_X\left(\widetilde{\CNOT}^{\otimes m}\right)_{Y:F_x}\ket{\psi_x}_{YEF}
\]
        with $\widetilde{\CNOT}_{A:B}=H_A\CNOT_{A:B}H_A$, and
        therefore
	\begin{align*}
          \Bigl[N_F,U_{XYF}\Bigr]\ket x_X\ket{\psi_x}_{YEF}
          &=\ket x_X\left[N_F,\left(\widetilde{\CNOT}^{\otimes m}\right)_{Y:F_x}\right]\ket{\psi_x}_{YEF}\nonumber\\
                                                            &=\ket x_X\left[(\one-\proj 0)_{F_x},\left(\widetilde{\CNOT}^{\otimes m}\right)_{Y:F_x}\right]\ket{\psi_x}_{YEF}.
	\end{align*}
	It follows that
	\begin{align}
          \Bigl\|\Bigl[N_F,&U_{XYF}\Bigr]\ket\psi_{XYEF}\Bigr\|_2\\
          &=\sum_{x \in \bits^n}p(x)\left\|\left[N_F,U_{XYF}\right]\ket{\psi_x}_{YEF}\right\|_2\nonumber\\
                                                                  &=\sum_{x \in \bits^n}p(x)\left\|\left[(\one-\proj 0)_{F_x},\left(\widetilde{\CNOT}^{\otimes m}\right)_{Y:F_x}\right]\ket{\psi_x}_{YEF}\right\|_2\nonumber\\
                                                                  &\le \left\|\left[(\one-\proj 0)_{F_{0^n}},\left(\widetilde{\CNOT}^{\otimes m}\right)_{Y:F_{0^n}}\right]\right\|_\infty,
	\end{align}
	where we have used the definition of the operator norm and the normalization of $\ket\psi_{XYEF}$ in the last line. For a unitary $U$ and a projector $P$, it is easy to see that $\|[U,P]\|_\infty\le 1$, as
$
		[U,P]=PU(\one-P)-(\one-P)UP
$
	is a sum of two operators that have orthogonal support and singular values smaller or equal to $1$. We therefore get 
$
	\left\|\left[N_F,U_{XYF}\right]\ket\psi_{XYEF}\right\|_2 \le 1,
$
	and as the state $\ket\psi$ was arbitrary, this implies
$
	\bigl\|\left[N_F,U_{XYF}\right]\bigr\|_\infty\le 1.
$
	The example from equation \eqref{eq:example1} shows that the above is actually an equality. The observation that $P_l\eta_F=0$ for all $l>0$ and an induction argument proves the second statement of the lemma.
\end{proof}

\subsection{A simulation theorem: The effect of random blinding}\label{sec:simulation-theorem-proof}

We now restate \expref{Theorem}{thm:blinded-algo} and provide a full proof.
\randomblinding*

\begin{proof}
	For a function $Q: \{0,1\}^n \rightarrow \{0,1\}^m$, we let $\mathcal{O}_Q$ denote the unitary map $\ket{x}\ket{y} \mapsto \ket{x}\ket{y \oplus Q(x)}$. Recall that $\algo A$ is specified by a fixed initial state $\ket{\phi_0}$ in some finite-dimensional Hilbert space, a sequence of $T$ unitary ``computation'' operators  $C_1, \ldots, C_k$, and a POVM $\{P_i : i \in I\}$. The distribution (on $I$) resulting from the algorithm applied to the oracle $\mathcal O_Q$ is given by applying the POVM to the state
$
	\ket{\phi^{Q}} := C_T \mathcal{O}_Q C_{T-1} \cdots \mathcal{O}_Q C_0 \ket{\phi_0}\,.
$
	Recall that if the trace distance between two such states satisfies
$
	\delta\bigl(\ket{\phi^{Q_1}}, \ket{\phi^{Q_2}}\bigr) := \sqrt{1 - |\langle\phi^{Q_1}| \phi^{Q_2}\rangle|^2} \leq \epsilon
$
	then the distance in total variation between the distributions produced by \emph{any} POVM on these two states is no more than $\epsilon$. In our case, we are interested in controlling $\Exp_B[\delta(\phi^F, \phi^{P \oplus F})].$ Define $F'=F\oplus P$. In preparation for a standard hybrid argument, define
$$
\ket{\phi_k} = \underbrace{C_T \mathcal{O}_{F'} \cdots \mathcal{O}_{F'}}_{(\dagger)} C_k \underbrace{\mathcal{O}_{F} \ldots \mathcal{O}_F C_{0}}_{(\ddag)} \ket{\phi_0} 
\qquad
\ket{\phi_k^F} = C_k \underbrace{\mathcal{O}_{F} \ldots \mathcal{O}_F C_{0}}_{(\ddag)} \ket{\phi_0}\,,
$$
	so that all oracle invocations in $(\dagger)$ are answered according to $\mathcal{O}_{F'}$ and all those in $(\ddag)$ are answered according to $\mathcal{O}_{F}$. Since $\delta$ is a metric on pure states, we have
	\[
	\Exp \delta(\ket{\phi^F}, \ket{\phi^{P \oplus F}}) \leq \Exp \sum_{k=1}^T \delta(\ket{\phi_k}, \ket{\phi_{k-1}}) = \sum_{k=1}^T \Exp \delta(\ket{\phi_k}, \ket{\phi_{k-1}})\,.
	\]
Note that $\delta$ is invariant under (simultaneous) unitary action, and hence for any $F$, $B$, and $P$,
\begin{align*}
	 &{\phantom{=~}}\delta(\ket{\phi_k}, \ket{\phi_{k-1}}) \nonumber\\
	 & = \delta(C_T \mathcal{O}_{F'} \cdots \mathcal{O}_{F'} C_k \mathcal{O}_{F} \ldots \mathcal{O}_F C_{0} \ket{\phi_0}, C_T \mathcal{O}_{F'} \cdots \mathcal{O}_{F'} C_{k-1} \mathcal{O}_{F} \ldots \mathcal{O}_F C_{0} \ket{\phi_0})\\
	& = \delta(\mathcal{O}_{F} C_{k-1} \ldots \mathcal{O}_F C_{0} \ket{\phi_0},   \mathcal{O}_{F'} C_{k-1} \ldots \mathcal{O}_F C_{0} \ket{\phi_0})\\
	& = \delta(\mathcal{O}_{F} \ket{\phi^{F}_{k-1}},   \mathcal{O}_{F'} \ket{\phi^{F}_{k-1}})
	=  \delta(\ket{\phi^{F}_{k-1}},   \mathcal{O}_F \mathcal{O}_{F'} \ket{\phi^{F}_{k-1}})
	=  \delta(\ket{\phi^{F}_{k-1}},   \mathcal{O}_{P} \ket{\phi^{F}_{k-1}})\,.
\end{align*}
For pure states $\ket{\psi}$ and $\ket{\psi'}$,
$\delta(\ket{\psi},\ket{\psi'}) \leq \| \ket{\psi} - \ket{\psi'} \|$.
Note that
$\ket{\psi} = \Pi_B \ket{\psi} + (I - \Pi_B) \ket{\psi}$, and $O_P$
operates identically on $(I - \Pi_B) \ket{\psi}$. Therefore
\begin{align*}
  \Exp_B[\delta(\phi^F, \phi^{P \oplus F})] &\leq T
                                              \max_{\ket{\phi}}
                                              \Exp \|
                                              \ket{\phi} -
                                              \calO_P\ket{\phi}\| \\
                                            & = T\max_{\ket{\phi}} \Exp_B \|\Pi_B
                                              \ket{\phi}  -
                                              \calO_P \Pi_B
                                              \ket{\phi} + (1
                                              - \calO_P)(I -
                                              \Pi_B)\ket{\phi} 
                                              \| \\
                                            & \leq
                                              T \max_{\ket{\phi}}
                                              \Exp_B
                                              (\|\Pi_B\ket{\phi}\|
                                              +
                                              \|\calO_P \Pi_B
                                              \ket{\phi}\|) \\
                                            & = 2 T \max_{\ket{\phi}}
                                              \Exp_B \|\Pi_B
                                              \ket{\phi}\| \\
                                            & \leq
                                              2T
                                              \max_{\ket{\phi}}\sqrt{\Exp_B
                                              |\bra{\phi}
                                              \Pi_B
                                              \ket{\phi}|}
                                              \quad
                                              \text{(Jensen's
                                              inequality)}\, .
        \end{align*}
        Let $\pi$ be a uniformly random element of the symmetric group
        on $\{0,1\}^n$ and $U_\pi$ be the unitary operator associated
        with the permutation $\pi$.  We have that
        \[ \Exp_B \Bigl[|\langle \phi | \Pi_{B} |\phi\rangle|\Bigr] =
          \Exp_B \Exp_\pi \Bigl[|\langle \phi| U_\pi \Pi_{B}
          U_\pi^{-1} |\phi\rangle|\Bigr] =2^{-n} \Exp_B[\tr{\Pi_B}] =
          \epsilon\, . \]
        Thus we conclude that
        $
        \Exp_B[\delta(\phi^F, \phi^{P \oplus F})] \leq 2T\sqrt \veps
        $.
        \end{proof}

\section{SUPPLEMENTARY MATERIAL}
\subsection{Non-adaptive quantum queries and ``double spending''}\label{sec:double-spending}

The following lemma shows that if there exists a non-adaptive quantum algorithm $\adver$ making $q$ queries to a function $f:\{0,1\}^n\to \{0,1\}^m$ that learns a certain property $p(f)$, then with inverse polynomial probability, there exists another non-adaptive $q$-query algorithm that learns $p(f)$ and $q$ input-output-pairs with inverse polynomial probability. For this to hold, we need to assume that $\adver$ makes its queries using a blank output register (i.e., initialized in the $\ket{0}$ state). This is the case, e.g., in period-finding and Simon's algorithm.

In the following, denote the set of $n$-bit-to-$m$-bit functions by $\mathcal F(n,m)$.
\begin{lemma}[Double spending lemma]\label{lem:double-spending}
Let $F\subseteq \mathcal F(n,m)$ be a set of functions, $P$ a set, $p: F\to P$ a function, and $D$ a probability distribution on $F$. Suppose there exists a quantum query algorithm $\adver$ which makes $q$ non-adaptive quantum queries to $\mathcal O_f$ with blank output register for $f \from D$ and outputs $p(f)$ with $1/\poly$ probability. Then there also exists an algorithm $\adver'$ which makes $q$ non-adaptive quantum queries to $\mathcal O_f$ for $f \from D$ and outputs both $p(f)$ and $q$ input-output pairs of $f$ with $1/\poly$ probability.
\end{lemma}
\begin{proof}
  Let $\mathcal X=\{0,1\}^n$, $\mathcal Y=\{0,1\}^m$ and
  $\Hi_Z=\mathbb C \mathcal Z$ for $Z=X,Y$. Set $\adver^\mathcal O (1^n)=\mathcal E\left( \mathcal O^{\otimes q}\ket
    \psi_{X^q}\otimes \ket 0_{Y^q}\right)$ where $\ket \psi$ is some
  input state and $\mathcal E=\{E_p\}_{p\in P}$ is a POVM on
  $\Hi_X^{\otimes q}\otimes \Hi_Y^{\otimes q}$ with outcomes labelled
  by the possible properties of $f$. Let
  $\ket{\psi_1}=\mathcal O^{\otimes q}\ket \psi_{X^q}\otimes \ket
  0_{Y^q}$. $\adver$ outputs $p(f)$ with inverse polynomial
  probability, say with probability
  $p_{\mathrm{succ}}=\bra{\psi_1}E_{p(f)}\ket{\psi_1}$. It follows that
  the post-measurement state conditioned on the outcome
  $p(f)$,
  $$\ket{\psi_2^{p(f)}}=\frac{\sqrt{E_{p(f)}}\ket{\psi_1}}{\sqrt{\bra{\psi_1}E_{p(f)}\ket{\psi_1}}},$$
  has inverse polynomial overlap with $\ket{\psi_1}$,
	\begin{align}
	\braket{\psi_1}{\psi_2^{p(f)}}&=\frac{\bra{\psi_1}\sqrt{E_{p(f)}}\ket{\psi_1}}{\sqrt{\bra{\psi_1}E_{p(f)}\ket{\psi_1}}}\nonumber\\
	&\ge\sqrt{\bra{\psi_1}E_{p(f)}\ket{\psi_1}}.
	\end{align}
	This implies immediately that measuring $\ket{\psi_2^{p(f)}}$ in the computational basis will yield $q$ input output pairs of $f$ with inverse polynomial success probability.
\end{proof}

We remark that the distribution of input-output pairs is at most $1 - 1/\poly$ far from the distribution one would get by simply measuring immediately after the query of $\algo A$. This means that, in the case of period-finding and Simon's algorithm (where the queries are uniform), the input-output pairs will be distinct with non-negligible probability.

\subsection{Alternative proof that random functions are \BZ-secure}

Using \expref{Lemma}{lem:Zhancrement}, we can give a simple proof of the fact that a random function is \BZ-secure. Because of its simplicity, and because much of it can be reused to prove a separation between \BZ and \BU, we provide this proof below.

\begin{theorem}[\cite{BZ13a}]\label{thm:Rand-BZ}
	An algorithm making $q$ quantum queries to a random oracle $f:\bits^n\to\bits^m$ produces $q+1$ input-output pairs of $f$ with probability at most
	\begin{equation}
		\frac{2^{\lceil\log(q+1)\rceil}}{2^m}.
	\end{equation}
\end{theorem}
Note that the probability bound is within a factor of $2$ of the one obtained in \cite{BZ13a}, and matches it for $q+1=2^k$, $k\in\N$.
\begin{proof}
	Let $\adver$ be an adversary that, when provided with the quantum random oracle $f$, outputs $q+1$ candidate input-output pairs. Formally, let $\rho_{(X,Y)^{q+1}F}$ be the joint cq-state of the adversary and the FO, where the classical registers $(X,Y)^{q+1}$ contain $\adver$'s output and $F$ is the FO's register. If we wanted to determine the success of $\adver$ at this point, we would apply the Fourier transform to $F$, and then measure $F$ and check if the outcome for $F_{x_i}$ is $y_i$ for each $(x_i,y_i)$ output by $\adver$.  
	
	Note that $P_l\rho=0$ for all $l>q$ by \expref{Lemma}{lem:Zhancrement}, i.e., there are at most $q$ entries of $F$ that are nonzero. This implies that the entry corresponding to at least one of the inputs that $\adver$ has output is, in fact, equal to $0^m$. However, this is only true in superposition: different branches of the superposition may have different entries in the state $\ket{0^m}$. We will deal with this issue by thinking about a new algorithm $\algo B$, which will simulate the entire execution of $\adver$ (including the oracle) and then perform a small number of additional measurements prior to the success check. The additional measurements will find a pair $(x_{i_0},y_{i_0})$ in $(X, Y)^{q+1}$ such that $F_{x_{i_0}}$ is actually in the state $\ket{0^m}$ (in every branch of the superposition). The probability that $y_{i_0} = f(x_{i_0})$ (in the execution of $\algo B$) will then be $2^{-m}$. We will then apply \expref{Lemma}{lem:add-measure} to show that the success probability of $\adver$ is not much better.
	
	 We now describe $\algo B$ in detail. Initially, $\algo B$ simulates both $\adver$ and the oracle. After $\adver$ has finished, but before the success check is performed, $\algo B$ (which is in the state $\rho$) applies binary search to the $q+1$ inputs that $\adver$ has output. The goal is to find an input $x_{i_0}$ such that $F_{x_{i_0}}$ is in state $\ket{0^m}$. We do this using binary measurements that ask ``are any of the registers $F_{x_{i_1}},...,F_{x_{i_k}}$ in the state $\ket{0^m}$?'' We split up the set $S_0=\{x_1,...,x_{q+1}\}$ into two subsets $S_0^\textsf{L}=\{x_1,...,x_{\lfloor  (q+1)/2\rfloor}\}$ and $S_0^\textsf{R}=\{x_{\lfloor (q+1)/2\rfloor+1},...,x_{q+1}\}$, and measure whether $F_x$ is in a state different from $\ket{0^n}$ for all $x\in S_0^\textsf{L}$. This is done using the binary measurement given by
	 \begin{equation}
	 	P_1=(\one-\proj 0)_{x_1}\otimes ...\otimes(\one-\proj 0)_{x_{\lfloor  (q+1)/2\rfloor}}\otimes \one^{\otimes(2^n-\lfloor  (q+1)/2\rfloor)}
	 \end{equation} 
	 and its complementary projector $P_0=\one-P_1$.
	 If the outcome is no, we set $S_1=S_0^\textsf{L}$, if it is yes then we set $S_1=S_0^\textsf{R}$. This makes sure that we continue with a set that contains an input such that the corresponding FO register is in state $\ket{0^n}$. Now we repeat the described steps using $S_1$ in place of $S_0$ and continue recursively until we encounter a set $S_l$ with only one element, say $w$. Continuing with the success check, we now know that $F_w$ is in the state $\ket{0^m}$, which implies that $f(w)$ is uniformly random and independent of $\adver$'s output. Indeed, a register that is in a pure state is automatically in product with the rest of the universe, and $f(w)$ is determined by applying $H^{\otimes m}$, which transforms $\ket{0^m}$ into $\ket{\phi_0}$, and measuring, which yields a uniformly random outcome. Therefore $\adver$'s success probability is at most $2^{-m}$. The total number of binary measurements for the binary search procedure is upper-bounded by $\lceil\log(q+1)\rceil$, so an application of \expref{Lemma}{lem:add-measure} finishes the proof.
\end{proof}

\subsection{A useful lemma on \bfilter}\label{app:more-bfilter}

Recall that blinding a function $f: \bits^n \to \bits^t$ on a set
$B \subseteq \bits^n$ results in the blinded function $Bf$ defined by
$B f (x) = \bot = (0^t, 1)$ for $x \in B$ and $B f (x) = (f(x), 0)$
for $x \notin B$. The following lemma, which is implicit in the proof
of Hash-and-MAC construction, could be useful in security reductions
involving \bfilter. 

\begin{lemma}
Let $h: \bits^n \to \bits^m$ be a \bfilter and $f:\bits^n \to \bits^t$ an efficiently computable function. Then for all oracle QPTs $(\algo A, \epsilon)$, we have
$$
\left|\Pr_{B \from \mathcal O_\epsilon} \left[ \algo A^{B f}(1^n) = 1\right]
- \Pr_{B \from \mathcal O_\epsilon^h} \left[ \algo A^{B f}(1^n) = 1\right]\right|
\leq \negl\,.
$$
\end{lemma}
\begin{proof}
It suffices to observe that one can simulate the oracle for $B f$ using two calls to an oracle for $\chi_B$ and two executions of $f$, as follows.
\begin{align*}
\ket{x}\ket{y} \ket b
&\mapsto \ket{x}\ket{y}\ket b\ket{\chi_B(x)}\ket{f(x)}\\
&\mapsto \ket{x}\ket{y \oplus \chi_B(x) \cdot f(x)}\ket{b\oplus \chi_B(x)}\ket{\chi_B(x)}\ket{f(x)}\\
&\mapsto \ket{x}\ket{y \oplus \chi_B(x) \cdot f(x)}\ket{b\oplus \chi_B(x)}\\
&= \ket{x}\ket{y \oplus B f(x)}.
\end{align*}
In the second step, we applied the CCNOT (Toffoli) gate to the second
register, with the fourth and fifth register as the controls and a
CNOT to the third register with the fourth register as a control. With
this observation, it is straightforward to turn any distinguisher for
$B_\epsilon f$ vs. $B_\epsilon^h f$ into one for $\chi_{B_\epsilon}$
vs.  $\chi_{B_\epsilon^h}$.
\end{proof}

\end{document}